%% file: simple-svp-nphard.tex
\documentclass{article}

\usepackage[utf8]{inputenc}
\usepackage{hdb_macros}
\usepackage{fullpage}
\usepackage{microtype}
\usepackage{enumitem}
\usepackage{thm-restate}
\usepackage{epigraph}
\usepackage[hyperpageref]{backref}

\newcommand{\RS}{\mathrm{RS}}
\newcommand{\lop}{\ensuremath{\lambda_1^{(p)}}}
\newcommand{\fh}{\ensuremath{\hat{f}}}
\newcommand{\gh}{\ensuremath{\hat{g}}}

\newcommand{\B}{\mathcal{B}}
\newcommand{\latpar}{\lat^{\perp}}
\renewcommand{\C}{\mathcal{C}}

\DeclareMathOperator{\Ima}{Im}

\usepackage[draft,multiuser,inline,nomargin]{fixme}
\FXRegisterAuthor{c}{ec}{\color{red}Chris}
\FXRegisterAuthor{h}{eh}{\color{blue}Huck}
\fxusetheme{color}

\title{Hardness of the (Approximate) Shortest Vector Problem:\\
A Simple Proof via Reed-Solomon Codes}
\author{Huck Bennett\thanks{Oregon State University, \email{huck.bennett@oregonstate.edu}. Part of this work was completed while the author was at the University of Michigan and supported by NSF Grant No.~CCF-2006857. The views
expressed are those of the authors and do not necessarily reflect the official policy or position of the National Science Foundation.}
\and
Chris Peikert\thanks{University of Michigan and Algorand, Inc., \email{cpeikert@umich.edu}. This author's work was supported by NSF Grant No.~CCF-2006857.}}
\date{\today}

\begin{document}

\maketitle
\listoffixmes

\begin{abstract}
\input{abstract}
\end{abstract}
\newpage %

\input{intro}
\input{prelims}

\input{local-density}
\input{decoding}
\input{derand}

\bibliography{svphard}
\bibliographystyle{alphaabbrvprelim}

\appendix
\input{app-svpcvp-reduction}
\input{app-fourier}

\end{document}

%% file: abstract.tex
We give a \emph{simple} proof that the (approximate, decisional) Shortest Vector Problem is $\NP$-hard under a randomized reduction.
Specifically, we show that for any $p \geq 1$ and any constant $\gamma < 2^{1/p}$, the $\gamma$-approximate problem in the $\ell_p$ norm ($\gamma$-$\GapSVP_p$) is not in $\RP$ unless $\NP \subseteq \RP$.
Our proof follows an approach pioneered by Ajtai (STOC 1998), and strengthened by Micciancio (FOCS 1998 and SICOMP 2000), for showing hardness of $\gamma$-$\GapSVP_p$ using \emph{locally dense lattices}.
We construct such lattices simply by applying ``Construction A'' to Reed-Solomon codes with suitable parameters, and prove their local density via an elementary argument originally used in the context of Craig lattices.

As in all known $\NP$-hardness results for $\GapSVP_p$ with $p < \infty$, our reduction uses randomness. Indeed, it is a notorious open problem to prove $\NP$-hardness via a deterministic reduction.
To this end, we additionally discuss potential directions and associated challenges for derandomizing our reduction.
In particular, we show that a close deterministic analogue of our local density construction would improve on the state-of-the-art explicit Reed-Solomon list-decoding lower bounds of Guruswami and Rudra (STOC~2005 and IEEE Transactions on Information Theory~2006).

As a related contribution of independent interest, we also give a polynomial-time algorithm for decoding $n$-dimensional ``Construction~A Reed-Solomon lattices'' (with different parameters than those used in our hardness proof) to a distance within an $O(\sqrt{\log n})$ factor of Minkowski's bound. This asymptotically matches the best known distance for decoding near Minkowski's bound, due to Mook and Peikert (IEEE Transactions on Information Theory 2022), whose work we build on with a somewhat simpler construction and analysis.

%% file: intro.tex
\section{Introduction}
\label{sec:intro}

\setlength{\epigraphwidth}{0.8\textwidth}

\epigraph{%
[I]t may easily happen that other, perhaps in some sense simpler, lattices also have the properties that are required from $L$ to complete the proof\ldots{}
There are different reasons which may motivate the search for such a lattice: to make the proof deterministic; to improve the factor in the approximation result; to make the proof simpler.}{
Mikl\'{o}s~Ajtai,~\cite[Remark~2]{DBLP:conf/stoc/Ajtai98}}

A \emph{lattice}~$\lat$ is the set of all integer linear combinations of some~$n$ linearly independent vectors $\vec{b}_1, \ldots, \vec{b}_n \in \R^m$. The matrix $B = (\vec{b}_1, \ldots, \vec{b}_n)$ whose columns are these vectors is called a \emph{basis} of~$\lat$, and~$n$ is called its \emph{rank}.
Formally, the lattice~$\lat$ generated by~$B$ is defined as
\[
\lat = \lat(B) := \set[\Big]{ \sum_{i=1}^n a_i \vec{b}_i : a_1, \ldots, a_n \in \Z } \ \text{.}
\]
Lattices are classically studied mathematical objects, and have proved invaluable in many computer science applications, especially the design and analysis of cryptosystems.
Indeed, the area of lattice-based cryptography, which designs cryptosystems whose security is based on the apparent intractability of certain computational problems on lattices, has flourished over the past quarter century. (See~\cite{journals/fttcs/Peikert16} and its bibliography for a comprehensive summary and list of references.)

The central computational problem on lattices is the Shortest Vector Problem ($\SVP$): given a lattice basis~$B$ as input, the goal is to find a shortest non-zero vector in $\lat(B)$.
This paper is concerned with its $\gamma$-approximate decision version in the~$\ell_p$ norm ($\gamma$-$\GapSVP_p$), where $p \geq 1$ is fixed and the approximation factor $\gamma = \gamma(n) \geq 1$ is some function of the lattice rank~$n$ (often a constant).
Here the input additionally includes a distance threshold $s > 0$, and the goal is to determine whether the length (in the~$\ell_p$ norm) $\lop(\lat) := \min_{\vec{v} \in \lat \setminus \set{\vec{0}}} \norm{\vec{v}}_p$ of the shortest non-zero vector in~$\lat$ is at most~$s$, or is strictly greater than $\gamma s$, when one of the two cases is promised to hold.
For the exact problem, where $\gamma = 1$, we often simply write $\GapSVP_p$.

Motivated especially by its central role in the security of lattice-based cryptography, understanding the complexity of $\gamma$-$\GapSVP$ has been the subject of a long line of work.
In an early technical report, van Emde Boas~\cite{vanEmdeBoas81} initiated the study of the hardness of lattice problems more generally, and in particular showed that $\GapSVP_{\infty}$ is $\NP$-hard.
Seventeen years later, Ajtai~\cite{DBLP:conf/stoc/Ajtai98} finally showed similar hardness for the important Euclidean case of $p = 2$,
i.e., he showed that exact $\GapSVP_2$ is $\NP$-hard, though under a \emph{randomized} reduction.
Subsequent work~\cite{conf/coco/CaiN98,journals/siamcomp/Micciancio00,journals/jcss/Khot06,journals/jacm/Khot05,journals/toc/HavivR12,journals/toc/Micciancio12} improved this by showing that $\gamma$-$\GapSVP_p$ in any~$\ell_p$ norm is $\NP$-hard to approximate for any constant $\gamma \geq 1$, and hard for nearly polynomial factors $\gamma = n^{\Omega(1/\log \log n)}$ assuming stronger complexity assumptions, also using randomized reductions.
Recent work~\cite{conf/stoc/AggarwalS18,conf/innovations/BennettPT22} has also shown the \emph{fine-grained} hardness of $\gamma$-$\GapSVP_p$ for small constants~$\gamma$ (again under randomized reductions).
On the other hand, $\gamma$-$\GapSVP_p$ for finite~$p \geq 2$ is unlikely to be $\NP$-hard for approximation factors $\gamma \geq C_p \sqrt{n}$ (where~$C_p$ is a constant depending only on~$p$)~\cite{journals/jcss/GoldreichG00,journals/jacm/AharonovR05,peikert08:_limit_hardn_of_lattic_probl}, and the security of lattice-based cryptography relies on the conjectured hardness of $\GapSVP$ or other problems for even larger (but typically polynomial) factors.

While this line of work has been very successful in showing progressively stronger hardness of approximation and fine-grained hardness for $\gamma$-$\GapSVP_p$, it leaves some other important issues unresolved.
First, for $p \neq \infty$ the hardness reductions and their analysis are rather complicated, and second, they are randomized.
Indeed, it is a notorious, long-standing open problem to prove that $\GapSVP_p$ is $\NP$-hard, even in its exact form, under a deterministic reduction for some finite~$p$.
While there have been some potential steps in this direction~\cite{journals/siamcomp/Micciancio00,journals/toc/Micciancio12}, e.g., using plausible number-theoretic conjectures that appear very hard to prove, there has been no new progress on this front for a decade.

\subsection{Our Contributions}
\label{sec:contributions}

The primary contribution of this work is to give a substantially simpler proof that $\gamma$-$\GapSVP_p$ is $\NP$-hard under a randomized reduction, for any $p \geq 1$ and constant $\gamma < 2^{1/p}$.
The heart of our reduction is a family of ``gadget'' lattices~$\lat$ derived from Reed-Solomon codes $\C \subseteq \F_q^n$ (for prime~$q$) via the very natural ``Construction~A''~\cite{conway1999sphere}, which simply sets $\lat = \C + q \Z^n$.
These lattices also are closely related to a family studied by Craig~\cite{Craig}, and we take advantage of this similarity in our analysis (see \cref{sec:related} for details).

\begin{theorem}[{Hardness of $\gamma$-$\GapSVP_p$}]
\label{thm:svp-hardness}
For any $p \geq 1$ and constant~$\gamma$ satisfying $1 \leq \gamma < 2^{1/p}$, $\gamma$-$\GapSVP_p$ is not in $\RP$ unless $\NP \subseteq \RP$.
\end{theorem}

We note that \cref{thm:svp-hardness} is actually identical to the main result in~\cite{journals/siamcomp/Micciancio00}.
As such, it matches the best known $\NP$-hardness of approximation for $\gamma$-$\GapSVP_p$ (i.e., largest $\gamma$) achieved by a ``one-shot'' reduction for all sufficiently small~$p$, including $p = 2$.
By ``one-shot,'' we mean that the reduction does not amplify the approximation factor from an initial fixed constant to an arbitrary constant (or more) via tensoring, as is done in~\cite{journals/jacm/Khot05,journals/toc/HavivR12,journals/toc/Micciancio12}.
(It is an interesting question whether our hard $\gamma$-$\GapSVP_p$ instances are amenable to tensoring; see \cref{sec:open-questions}.)

Although our reduction still uses randomness, we believe that it may be easier to derandomize than previous reductions, both due to its simplicity, and because of its close connection to prior work showing hardness of minimum distance problem on codes via a deterministic reduction~\cite{journals/tit/ChengW12}.
To that end, in \cref{sec:derand} we also describe two approaches to potentially derandomizing our reduction, both of which aim to deterministically construct a particular lattice coset and lower bound the number of short vectors in it (see \cref{sec:techniques} for the motivation for this).
The first approach is based on Fourier analysis, using similar techniques to those in~\cite{journals/tit/ChengW12}, and the second is based on ``smooth'' proxies for point-counting functions.

We also show that a close deterministic analog of our randomized local-density construction would imply improved explicit Reed-Solomon list-decoding lower bounds, going beyond the current state of the art from~\cite{DBLP:journals/tit/GuruswamiR06}.
One may interpret this implication either pessimistically, as a barrier to a very strong derandomization of our reduction, or optimistically, as a potential route to improve Reed-Solomon list-decoding lower bounds.
Here there is a further connection between the two problems, in that~\cite{DBLP:journals/tit/GuruswamiR06} obtains its list-decoding lower bounds by using the same Fourier-analytic tool underlying one of our derandomization attempts---specifically, the Weil bound for character sums (\cref{eq:weil-additive}).
Unfortunately, the Weil bound falls just short of what we need in our context.
(The Weil bound and related techniques were first used for counting Reed-Solomon code words in~\cite{conf/focs/ChengW04}, and were also used in the deterministic hardness reduction for the minimum distance problem on codes in~\cite{journals/tit/ChengW12}.)

\paragraph{Efficient decoding near Minkowski's bound.}

As a separate contribution of independent interest, in \cref{sec:decoding} we give a polynomial-time algorithm for decoding ``Construction~A Reed-Solomon lattices'' of rank~$n$---the same family of lattices as in our hardness reduction, but instantiated with different parameters---to a distance within a $O(\sqrt{\log n})$ factor of Minkowski's bound.\footnote{Minkowski's bound gives an upper bound on the ``normalized density'' of a lattice $\lat$.
Specifically, it asserts that $\lambda_1^{(2)}(\lat) \leq \sqrt{n} \cdot \det(\lat)^{1/n}$ for all rank-$n$ lattices $\lat$, where $\det(\lat) = \sqrt{\det(B^T B)}$ for any basis~$B$ of~$\lat$.}
The $O(\sqrt{\log n})$ factor in our result asymptotically matches the best factor known from prior work~\cite{DBLP:journals/tit/MookP22}, which is for a different family of lattices.
In fact, we rely on one of the main underlying theorems from that work, but give a simpler construction and analysis based on individual Reed-Solomon codes instead of towers of BCH codes.

Let $\RS_q[k, S]$ denote the dimension-$k$ Reed-Solomon code over~$\F_q$ with evaluation set~$S$ (defined below in \cref{eq:intro-RS-def}).
Note that $n = q$ is the rank of the lattice~$\lat$ in the following theorem.

\begin{theorem}[{Decoding near Minkowski's bound, informal}] \label{thm:informal-decoding}
Let~$q$ be prime and let $k := \floor{ q/(2 \log_2 q) } \leq q/2$.
Then for the ``Construction~A Reed-Solomon'' lattice $\lat := \RS_q[q - k, \F_q] + q \Z^q \subseteq \Z^q$:

\begin{enumerate}
\item \label{item:informal-high-density} We have $\Omega(\sqrt{q/\log q}) \leq \lambda_1(\lat) \leq \sqrt{q} \cdot \det(\lat)^{1/q} \leq O(\sqrt{q})$, i.e., the minimum distance is within a $O(\sqrt{\log q})$ factor of Minkowski's bound.

\item \label{item:informal-efficient-decoding} 
There is an algorithm that, on input~$q$ and a vector $\vec{y} \in \R^q$, outputs all lattice vectors $\vec{v} \in \lat$ satisfying $\norm{\vec{y} - \vec{v}} \leq C \sqrt{k} \approx C \sqrt{q/(2 \log_2 q)}$ in time polynomial in~$q$, for some universal constant $C > 0$.
\end{enumerate}
\end{theorem}

This result adds to a separate line of work on efficient (list) decoding for various families of lattices~\cite{conf/isit/MicciancioN08,journals/cc/GrigorescuP17,journals/dcc/DucasP19,DBLP:journals/tit/MookP22}.
Recently, Ducas and van Woerden~\cite{cryptoeprint:2021:1332} further motivated this study by showing cryptographic applications of lattices that can be efficiently decoded near Minkowski's bound. (However, their application is most compelling when the minimum distances of both the lattice and its dual are close to Minkowski's bound, which is not the case in the present setting.)

\subsection{Technical Overview}
\label{sec:techniques}

Here we give an overview of the key new elements in the proof of our main hardness theorem (\cref{thm:svp-hardness}), which are the focus of \cref{sec:ldl-reed-solomon}.
For concision, we defer the technical aspects of our efficient decoding algorithm and derandomization attempts to \cref{sec:decoding,sec:derand}, respectively.

Besides using randomness, another common feature in nearly all prior hardness results for $\GapSVP_p$ is the use of \emph{locally dense lattices} as advice (the only exception being~\cite{journals/jcss/Khot06}).
Roughly speaking, a locally dense lattice for relative distance $\alpha \in (0, 1)$ in the~$\ell_p$ norm is a lattice~$\lat$ and a coset $\vec{x} + \lat$ (i.e., the lattice ``shifted by'' some vector~$\vec{x}$) such that there are at least subexponentially many (in the lattice rank) vectors $\vec{v} \in \vec{x} + \lat$ satisfying $\norm{\vec{v}}_p \leq \alpha \cdot \lop(\lat)$.
One may view such a coset $\vec{x} + \lat$ as a ``bad'' configuration for list-decoding~$\lat$ to within relative distance~$\alpha$ in the~$\ell_p$ norm, because there are many lattice vectors relatively close to~$-\vec{x}$.%
\footnote{For technical reasons, the formal definition of local density, given in \cref{def:locally-dense}, also requires a linear transform that maps the short vectors in $\vec{x}+\lat$ onto the set of all binary vectors of a given dimension.
Such a transform can be obtained by random sampling using a probabilistic version of Sauer's Lemma (see \cref{thm:prob-sauer}) that is now standard in this context~\cite{DBLP:conf/stoc/Ajtai98,journals/siamcomp/Micciancio00}.
Therefore, we defer further discussion of this issue to the main body.}

Prior works have obtained locally density from a variety of lattice families: the Schnorr-Adelman prime number lattices~\cite{DBLP:conf/stoc/Ajtai98,conf/coco/CaiN98,journals/siamcomp/Micciancio00}; a variant of Construction~A~\cite{journals/jacm/Khot05} and Construction~D~\cite{journals/toc/Micciancio12} applied to (towers of) BCH codes; and random sublattices of $\Z^n$ and lattices with exponential kissing number~\cite{conf/stoc/AggarwalS18,conf/innovations/BennettPT22}.
In this work, we give a simple construction of locally dense lattices from Reed-Solomon codes, as described below.

Our main reduction (\cref{thm:cvp-to-svp}) shows how to use a locally dense lattice for relative distance~$\alpha$ in the~$\ell_p$ norm to prove $\NP$-hardness (via a randomized reduction) of $\gamma$-$\GapSVP_p$ for any constant $\gamma > 1/\alpha$.
This reduction is very similar to those from prior works, so for the remainder of this section we focus on summarizing our new construction of locally dense lattices.

\paragraph{Locally dense lattices from Reed-Solomon codes.}

We start with some basic definitions and facts used in our construction. Recall that the Construction~A lattice obtained from a linear code $\C \subseteq \F_q^n$ for some prime~$q$ is defined as $\lat := \C + q \Z^n$, i.e., an integer vector $\vec{z} \in \Z^n$ is in the lattice if and only if $\vec{z} \bmod q$ is a code word.
In fact, it will often be convenient to work with an equivalent ``dual view'' of Construction~A lattices. 
Namely, if $H \in \F_q^{k \times n}$ is a parity-check matrix of a linear code $\C := \ker(H) \subseteq \F_q^n$ for prime $q$, then the \emph{parity-check lattice} $\latpar(H)$ obtained from~$H$ is defined as
\begin{equation}
\label{eq:intro-parity-check-lattice}
\latpar(H) := \set{\vec{z} \in \Z^n : H\vec{z} = \vec{0} \in \F_q^k}
= \ker(H) + q \Z^n = \C + q \Z^n \  \text{.}
\end{equation}
Such lattices have determinant $\det(\latpar(H)) = \abs{\Z^n/\latpar(H)} \leq q^k$, with equality exactly when~$H$ has full row rank (see \cref{lem:parity-lat-properties}).

We next define the family of parity-check matrices $H = H_{q}(k, S)$ that we use to construct our family of locally dense lattices.
Such a matrix is parameterized by a prime~$q$, a positive integer~$k$, and a set $S \subseteq \F_q$.
Letting $s_0, \ldots, s_{n-1}$ be the elements of $S$ in some arbitrary order, we define
\begin{equation}
\label{eq:parity-check-intro}
H = H_{q}(k, S) := \begin{pmatrix*}[l]
1 & 1 & 1 & \cdots & 1 \\
s_0 & s_1 & s_2 & \cdots & s_{n-1} \\
s_0^2 & s_1^2 & s_2^2 & \cdots & s_{n-1}^2 \\
\vdots & \vdots & \vdots & \ddots & \vdots \\
s_0^{k-1} & s_1^{k-1} & s_2^{k-1} & \cdots & s_{n-1}^{k-1}
\end{pmatrix*} \in \F_q^{k \times n} \ \text{.}
\end{equation}
That is, $H_q(k, S)$ is the transposed Vandermonde matrix whose $(i, s)$th entry is $s^i$, where for convenience we index the rows and columns of $H_q(k, S)$ by $i \in \set{0, \ldots, k - 1}$ and $s \in S$, respectively, and define $0^0 := 1$.

The matrix $H = H_q(k, S)$ defined in \cref{eq:parity-check-intro} is a \emph{generator matrix} of the dimension-$k$ Reed-Solomon code
\begin{equation}
    \label{eq:intro-RS-def}
    \RS_q[k,S] := \set{ (p(s))_{s \in  S} : p \in \F_q[x], \deg(p) < k}
\end{equation}
over~$\F_q$ with evaluation set~$S$, and hence is a parity-check matrix of its dual code, which is a so-called \emph{generalized} Reed-Solomon~(GRS) code (see~\cite[Theorem~5.1.6]{hall-coding}).
Moreover, in the special case where $S = \F_q$, it turns out that $H_q(k, S)$ is a parity-check matrix for the (ordinary) Reed-Solomon code $\RS_q[q - k, \F_q]$ of dimension $q - k$ with evaluation set $S = \F_q$. So, $\latpar(H_q(k, \F_q)) = \RS_q[q - k, \F_q] + q \Z^q$ is the Construction~A lattice corresponding to the Reed-Solomon code $\RS_q[q - k, \F_q]$.
For simplicity, in this overview we restrict to these ``Construction~A Reed-Solomon'' lattices by taking $S = \F_q$, but note that our results generalize to any sufficiently large $S \subseteq \F_q$.

It is easy to see that for $k < q$, the GRS code having parity-check matrix~$H$ has minimum distance (in the Hamming metric) $k+1$: any~$k$ columns of~$H$ are linearly independent, because they form a transposed Vandermonde matrix, while any $k+1$ obviously are not.
Therefore, the corresponding Construction~A lattice $\lat := \ker(H) + q \Z^q$ has minimum distance $\lambda_1^{(1)}(\lat) \geq k+1$ in the~$\ell_1$ norm.
The key to our local density construction and its $\alpha \approx 1/2^{1/p}$ relative distance is that the~$\ell_1$ minimum distance is in fact almost \emph{twice} this large (at least): in \cref{thm:rs-min-dist} we show that $\lambda_1^{(1)}(\lat) \geq 2k$ when $k \leq q/2$.
The proof is short and elementary, proceeding via Newton's identities, and is closely related to the analysis of Craig lattices~\cite{Craig} (see \cref{sec:related}).

\paragraph{Obtaining a dense coset.}

Because the determinant of~$\lat$ (i.e., the number of its integer cosets) is~$q^k$, the pigeonhole principle immediately implies that there exists an integer coset of~$\lat$ containing at least $\binom{q}{h} / q^k$ \emph{binary} vectors in $\bit^q$ with Hamming weight~$h$, which have~$\ell_1$ norm~$h$.
By setting parameters appropriately, this yields a coset with subexponentially many vectors of~$\ell_1$ norm at most $\alpha \cdot \lambda_1^{(1)}(\lat)$, for any constant $\alpha > 1/2$.

More specifically, set $h := \alpha \cdot (2k) \leq \alpha \cdot \lambda_1^{(1)}(\lat)$ and $q \approx k^{1/\eps}$ for some positive constant $\eps < 1-1/(2\alpha)$.
(For simplicity, assume that~$h$ is an integer.)
Then there must exist an integer coset of~$\lat$ containing at least
\begin{equation}
\label{eq:intro-pigeonhole}
\frac{\binom{q}{h}}{q^k} \geq \Big(\frac{q}{h}\Big)^{h} \cdot q^{-k} 
= \frac{q^{(2\alpha - 1)k}}{(2 \alpha k)^{2 \alpha k}} \approx \frac{q^{(2 \alpha - 1)k}}{q^{2 \eps \alpha k}} = q^{(2 (1 - \eps) \alpha - 1)k} = q^{\Omega(k)} = q^{\Omega(q^\eps)}
\end{equation}
weight-$h$ binary vectors, which is subexponentially large in~$q$.

The above shows the \emph{existence} of a suitable coset, but following previous works, it is straightforward to show that a \emph{randomly sampled} coset from a suitable distribution is likely to have enough short vectors (see \cref{lem:coset-sampling}).
Indeed, the difference between showing that such a coset exists, versus sampling one efficiently, versus deterministically computing one efficiently, is the main technical difference between getting a non-uniform, versus randomized, versus deterministic hardness reduction (respectively) for $\GapSVP_p$ using these techniques.

The above argument generalizes to arbitrary~$\ell_p$ norms for finite~$p$, albeit for larger relative distances~$\alpha > 1/2^{1/p}$.
Because~$\lat$ is integral, $\lambda_1^{(1)}(\lat) \geq 2k$ implies that $\lop(\lat) \geq (2k)^{1/p}$ for any finite $p \geq 1$.
Moreover, reparameterizing the calculation in \cref{eq:intro-pigeonhole} by choosing $\alpha > 1/2^{1/p}$ and setting $h := \alpha^p \cdot (2 k)$ shows that some coset of~$\lat$ contains subexponentially many binary vectors of Hamming weight~$h$, and hence of~$\ell_p$ norm $h^{1/p} = \alpha \cdot (2k)^{1/p} \leq \alpha \cdot \lop(\lat)$.
Therefore, this construction yields locally dense lattices in the~$\ell_p$ norm for any constant relative distance $\alpha > 1/2^{1/p}$, which by our main reduction implies \cref{thm:svp-hardness}, i.e., randomized $\NP$-hardness of $\gamma$-$\GapSVP_p$ for any constant $\gamma < 2^{1/p}$.

\subsection{Additional Related Work}
\label{sec:related}

The Construction~A Reed-Solomon lattices we use are closely related to a family of algebraic lattices studied by Craig~\cite{Craig} (see also~\cite[Chapter~8, Section~6]{conway1999sphere}), and our proof of \cref{thm:rs-min-dist} is inspired by and very similar to an argument lower bounding the minimum distance of Craig lattices from~\cite[Chapter~8, Section~6, Theorem~7]{conway1999sphere}.\footnote{Specifically,~\cite{conway1999sphere} considers Craig lattices obtained from the coefficient vectors of polynomials in principal ideals of the form $(x - 1)^m R$ in rings of the form $R = \Z[x]/(x^p - 1)$, for some prime~$p$ and integer $m \geq 1$. The original definition of~\cite{Craig} is slightly different, and uses the ``canonical'' (Minkowski) embedding of such ideals in the ring of integers $R=\Z[x]/(\Phi_p(x))$ of the $p$th cyclotomic number field.}
In fact, the guarantees on the minimum distances and determinants of Craig lattices given there are such that our hardness proof for $\GapSVP$ could be made to work using Craig lattices instead of Construction~A Reed-Solomon lattices.
We have chosen to work with the latter family because they are more elementary and general, and because of their direct connections to  efficient lattice decoding, list-decoding lower bounds, and derandomization techniques, as explored in \cref{sec:decoding,sec:derand}.

We also note that locally dense lattices appear in other work on the complexity of lattice problems. In particular, all prior $\NP$- and fine-grained hardness reductions for the problem of Bounded Distance Decoding to relative distance~$\alpha'$ in the~$\ell_p$ norm ($\alpha'$-$\BDD_p$)~\cite{conf/approx/LiuLM06,conf/coco/BennettP20,conf/innovations/BennettPT22}, another important lattice problem, use locally dense lattices.
In fact, the smallest~$\alpha'$ for which we know how to show $\NP$-hardness of $\alpha'$-$\BDD_p$ is essentially the smallest~$\alpha(p)$ for which we can construct locally dense lattices for relative distance~$\alpha(p)$ in the~$\ell_p$ norm.

Additionally, \emph{locally dense codes} play an analogous role to locally dense lattices in the complexity of coding problems, and particularly the Minimum Distance Problem (MDP), which is the analog of $\SVP$ for codes. Indeed, the work of Micciancio, Dumer, and Sudan~\cite{journals/tit/DumerMS03} adapted the local density framework to codes and used it to show randomized hardness of approximation for MDP. This reduction was subsequently made deterministic by Cheng and Wan~\cite{journals/tit/ChengW12}, who, as mentioned above, constructed locally dense codes from Reed-Solomon codes.
Micciancio~\cite{conf/coco/Micciancio14} formalized the locally dense code framework, gave a simplified deterministic $\NP$-hardness proof for MDP, and discussed the search for analogous deterministic constructions of locally dense lattices.
(Austrin and Khot~\cite{journals/tit/AustrinK14} also gave a simplified deterministic hardness reduction for MDP, but it does not obviously use the same framework.)

Finally, with respect to our result in \cref{thm:informal-decoding}, we note that a number of previous works have given algorithms for efficient decoding near Minkowski's bound on other families of lattices.
These include~\cite{conf/isit/MicciancioN08,journals/cc/GrigorescuP17}, which showed how to (list) decode to a distance within a $O(n^{1/4})$ factor of Minkowski's bound on Barnes-Wall lattices; 
\cite{journals/dcc/DucasP19}, which showed how to decode to a distance within a $O(\log n)$ factor of Minkowski's bound on a family of discrete-logarithm lattices; and~\cite{DBLP:journals/tit/MookP22}, which showed how to (list) decode to a distance within a $O(\sqrt{\log n})$ factor of Minkowski's bound on lattices obtained by applying Construction~D to towers of BCH codes.

\subsection{Open Questions}
\label{sec:open-questions}

The obvious question left open by our work is whether our reduction can be derandomized, using the same family of lattices.
Addressing this is the main focus of \cref{sec:derand}, and we defer further discussion of it to there.

Another interesting question is whether it is possible to amplify the approximation factor~$\gamma$ achieved by our reduction via \emph{tensoring}, as is done in previous works~\cite{journals/jacm/Khot05,journals/toc/HavivR12,journals/toc/Micciancio12}.
Recall that the tensor product of lattices $\lat_1, \lat_2$ having respective bases $B_1, B_2$ is defined as $\lat_1 \otimes \lat_2 := \lat(B_1 \otimes B_2)$, and that $\lambda_1(\lat_1 \otimes \lat_2) \leq \lambda_1(\lat_1) \cdot \lambda_1(\lat_2)$.
To see why tensoring might be useful for gap amplification, observe that if this bound was an equality, then the $m$-fold tensor product $\lat^{\otimes m}$ of~$\lat$ with itself would satisfy $\lambda_1(\lat^{\otimes m}) = \lambda_1(\lat)^m$, which would amplify the approximation factor in an instance of $\gamma$-$\GapSVP_p$ to $\gamma^m$.
However, this inequality is \emph{not} tight in general (see, e.g.,~\cite[Lemma~2.3]{journals/toc/HavivR12}).
Nevertheless, prior reductions have been able to prove similar inequalities for certain families of lattices, but it is not immediately obvious how to adapt their techniques to our setting.
Finally, as previously observed in~\cite{conf/coco/BennettP20}, we do not have a good understanding of the smallest possible relative distance $\alpha^* = \alpha^*(p)$ for locally dense lattices in~$\ell_p$ norms for finite $p > 2$.
Besides showing better ``one-shot'' hardness for $\gamma$-$\GapSVP_p$, achieving smaller factors $\alpha(p)$ for such~$p$ (ideally, $\alpha(p) \approx \alpha^*(p)$) would imply improved (randomized) $\NP$-hardness results for Bounded Distance Decoding (specifically, it would show hardness of $\alpha'$-$\BDD_p$ for any $\alpha' > \alpha(p)$).
The triangle inequality implies a lower bound of $\alpha^* \geq 1/2$ in any norm, and, as observed in~\cite{journals/siamcomp/Micciancio00}, we have $\alpha^*(2) \geq 1/\sqrt{2}$.
So, the relative-distance factors $\alpha(p) \approx 1/2^{1/p}$ for local density achieved by~\cite{journals/siamcomp/Micciancio00} and this work are essentially optimal for $p = 1$ and $p = 2$ (and conjecturally for all $p \in [1, 2]$).
However, they deteriorate as~$p$ increases.
On the other hand, the factors $\alpha(p)$ obtained in~\cite{conf/coco/BennettP20} satisfy $\lim_{p \to \infty} \alpha(p) = 1/2$, and hence are nearly optimal for large~$p$, but are poor even for some $p > 2$.
In particular, they satisfy
$\alpha(p) < 1/\sqrt{2}$ only for $p > p_1 \approx 4.2773$, even though we believe that $\alpha(p) < 1/\sqrt{2}$ should be achievable for all $p > 2$.
(We note in passing that $\alpha(p) \approx 1/\sqrt{2}$ is achievable for \emph{all} $p$ via norm embeddings~\cite{conf/stoc/RegevR06}.)
It seems plausible that a family similar to our Construction~A Reed-Solomon lattices $\lat$---perhaps obtained by parameterizing the underlying Reed-Solomon code differently, using a different algebraic code, or by randomly sparsifying $\lat$ as in~\cite{conf/coco/BennettP20}---might achieve this.

\paragraph{Acknowledgements.}

We thank Swastik Kopparty~\cite{kopparty_personal20} for very helpful answers to several of our questions, and for pointing us to the contents of \cref{sec:binary-rs-list,sec:fourier}.

%% file: prelims.tex
\section{Preliminaries}
\label{sec:prelims}

Throughout this work we adopt the convention that $0^0:=1$ in any ring.
For a positive integer~$k$, define $[k] := \set{0, 1, \ldots, k-1}$.

In general, every vector or matrix is indexed by some specified set~$S$. For example, $\vec{x} \in \Z^S$ is an integer vector indexed by~$S$, having an entry $x_s \in \Z$ for each $s \in S$ (and no other entries). When the index set is $[n]$ for some non-negative integer~$n$, we usually omit the brackets in the exponent and just write, e.g.,~$\Z^n$. We emphasize that in this case the indices start from zero. An object indexed by a finite set~$S$ of size $n=\card{S}$ can be reindexed by~$[n]$, simply by enumerating $S=\set{s_0, \ldots, s_{n-1}}$ under some arbitrary order, and identifying index~$s_i$ with index~$i$.

For a finite set~$S$ and a positive integer $h \leq \card{S}$, let $B_{S,h} := \set{\vec{v} \in \bit^S : \norm{\vec{v}}_1=h}$ be the set of binary vectors indexed by~$S$ of Hamming weight~$h$. As above, when $S=[n]$ we often write $B_{n,h}$.
Finally, let $\B_p^n(r) := \set{\vec{x} : \norm{\vec{x}}_p \leq r} \subset \R^n$ denote the real $n$-dimensional~$\ell_p$ ball of radius~$r$ centered at the origin.

\subsection{Basic Lattice Definitions}
\label{subsec:lattice-defs}

Given a lattice $\lat = \lat(B)$ with basis $B \in \R^{m \times n}$, we define the \emph{rank} of $\lat$ to be $n$ and the (ambient) \emph{dimension} of $\lat$ to be $m$.
We denote the \emph{minimum distance} of $\lat$ in the $\ell_p$ norm, which is the length of a shortest non-zero vector in $\lat$, by
\[
\lop(\lat) := \min_{\vec{x} \in \lat \setminus \set{0}} \norm{\vec{x}}_p \ \text{.}
\]
The central problem that we study in this work asks about the value of $\lop(\lat)$ for a given input lattice $\lat$.
\begin{definition} \label{def:gapsvp}
For $p \geq 1$ and $\gamma = \gamma(n) \geq 1$, the decisional, $\gamma$-approximate Shortest Vector Problem in the $\ell_p$ norm ($\gamma$-$\GapSVP_p$) is the promise problem defined as follows. The input consists of a basis $B \in \Z^{m \times n}$ of an integer lattice $\lat$ and a distance threshold $s > 0$, and the goal is to determine whether the input is a YES instance or a NO instance, where these are defined as follows:
\begin{itemize}
    \item YES instance: $\lop(\lat) \leq s$.
    \item NO instance: $\lop > \gamma s$.
\end{itemize}
\end{definition}

We define the \emph{determinant} of $\lat$ to be $\det(\lat) := \sqrt{\det(B^T B)}$, which is equal to $\abs{\det(B)}$ when $m = n$ (i.e., when $\lat$ is full-rank).
We note that determinant is well defined because, although lattice bases are not unique, they are equivalent up to multiplication on the right by unimodular matrices.

The density of a rank-$n$ lattice~$\lat$ is captured by its so-called \emph{root Hermite factor} $\lambda_1(\lat)/\det(\lat)^{1/n}$.\footnote{This ratio is the square root of the Hermite factor $\gamma(\lat) := (\lambda_1(\lat)/\det(\lat)^{1/n})^2$, which is defined in this way for historical reasons.}
The density of a lattice corresponds to its quality in various applications, including as the set of centers of a sphere packing and as an error-correcting code.
Minkowski's bound asserts that the root Hermite factor of such a rank-$n$ lattice is at most~$\sqrt{n}$, which is convenient to write in expanded form as
\begin{equation}
\label{eq:minkowski}
\lambda_1(\lat) \leq \sqrt{n} \cdot \det(\lat)^{1/n} \ \text{.}
\end{equation}

\subsection{Parity-Check Matrices and Lattices}
\label{sec:parity-check-matrices-lattices}

For a prime~$q$ and a matrix $H \in \F_q^{k \times n}$, we define the \emph{parity-check lattice} $\lat^{\perp}(H)$ obtained from~$H$ as
\begin{equation}
\label{eq:parity-check-lattice}
\latpar(H) := \set{\vec{z} \in \Z^n : H\vec{z} = \vec{0}} = \ker(H) + q \Z^n \  \text{.}
\end{equation}
Note that $\latpar(H)$ is simply the ``Construction A'' lattice~\cite[Chapter~5, Section~2]{conway1999sphere} of the linear error-correcting code~$\mathcal{C}$ having~$H$ as a parity-check matrix, i.e., $\mathcal{C} = \set{ \vec{c} \in \F_q^n : H \vec{c} = \vec{0}}$.
More generally, for any ``syndrome'' $\vec{u} \in \F_q^k$ we define
\[ \latpar_{\vec{u}}(H) := \set{\vec{x} \in \Z^n : H \vec{x} = \vec{u}} \ \text{.} \]
If there exists some $\vec{x} \in \Z^n$ such that $H \vec{x}=\vec{u}$, then it follows immediately that $\latpar_{\vec{u}}(H)$ is simply the lattice coset $\vec{x} + \latpar(H)$.
So, we can identify cosets of $\latpar(H)$ by their corresponding syndromes.
We recall some standard properties of parity-check lattices, and give a proof for self-containment.

\begin{lemma}%
\label{lem:parity-lat-properties}
Let $q$ be a prime, let~$k$ and~$n$ be positive integers, and let $H \in \F_q^{k \times n}$ be a parity-check matrix. Then the parity-check lattice $\lat = \latpar(H)$ has rank~$n$ and determinant $\det(\lat) \leq q^k$, with equality if and only if the rows of~$H$ are linearly independent.
\end{lemma}

\begin{proof}
The first claim follows simply by noting that $q \Z^n \subseteq \latpar(H) \subseteq \Z^n$.
For the determinant, observe that the map $\vec{x} \mapsto H\vec{x}$ is an additive-group homomorphism from~$\Z^n$ to~$\F_q^k$, and that $\latpar(H)$ is its kernel by definition.
So, by the first isomorphism theorem, the map induces an isomorphism from the quotient group $\Z^n/\latpar(H)$ to the image $\Ima(H) = \set{H \vec{x} : \vec{x} \in \Z^n} \subseteq \F_q^k$, where the subset relation is an equality if and only if the rows of~$H$ are linearly independent.
The claim then follows from the fact that $\det(\latpar(H)) = \card{\Z^n/\latpar(H)} = \card{\Ima(H)}$.
\end{proof}

We next formally define the family of parity-check lattices that are at the heart of our construction of locally dense lattices.

\begin{definition}
\label{def:parity-check-matrix}
For a prime~$q$, positive integer~$k$, and set $S \subseteq \F_q$, define $H_{q}(k, S) \in \F_q^{k \times S}$ to be the matrix~$H$ whose rows and columns are respectively indexed by~$[k] = \set{0, 1, \ldots, k-1}$ and~$S$, and whose $(i,s)$th entry is \[ H_{i,s}:=s^i \in \F_q \ \text{.} \]
(Recall that $0^0:=1$.)
Equivalently, if we enumerate $S=\set{s_0, \ldots, s_{n-1}}$ in some arbitrary order, we have
\begin{equation}
\label{eq:parity-check}
H = H_{q}(k, S) := \begin{pmatrix*}[l]
1 & 1 & 1 & \cdots & 1 \\
s_0 & s_1 & s_2 & \cdots & s_{n-1} \\
s_0^2 & s_1^2 & s_2^2 & \cdots & s_{n-1}^2 \\
\vdots & \vdots & \vdots & \ddots & \vdots \\
s_0^{k-1} & s_1^{k-1} & s_2^{k-1} & \cdots & s_{n-1}^{k-1}
\end{pmatrix*} \in \F_q^{k \times n} \ \text{.}
\end{equation}
\end{definition}
Notice that~$H$ is a transposed Vandermonde matrix.
In particular, if $k \leq n$ then its rows are linearly independent, and so $\det(\latpar(H)) = q^k$ by \cref{lem:parity-lat-properties}.

We recall from the introduction that $H = H_q(k,S)$ is a generator matrix (of row vectors) of the dimension-$k$ Reed-Solomon code over~$\F_q$ with evaluation set~$S$, and hence~$H$ is a parity-check matrix of its dimension-$(n-k)$ dual code.
So, $\latpar(H)$ is the Construction A lattice of this dual code.
These dual codes are in fact \emph{generalized Reed-Solomon codes}, a family of codes that include Reed-Solomon codes as a special case and that are closed under taking duals (see~\cite[Theorem~5.1.6]{hall-coding}).
Moreover, in the special case where $S = \F_q$, the matrix $H = H_q(k, \F_q)$ is in fact a parity-check matrix of an (ordinary) Reed-Solomon code.
For our hardness proof it suffices to use this special case; i.e., we show $\NP$-hardness (under a randomized reduction) of $\GapSVP$ by using Construction A lattices of (ordinary) Reed-Solomon codes as gadgets.
See \cref{sec:decoding} and \cref{sec:binary-rs-list} for other connections between these lattices and Reed-Solomon codes.

\subsection{Symmetric Polynomials}
\label{sec:symmetric-polynomials}

A \emph{symmetric polynomial} $P(x_1, x_2, \ldots, x_m)$ is a polynomial that is invariant under any permutation of its variables, i.e., $P(x_1, \ldots, x_m) = P(x_{\pi(1)}, \ldots, x_{\pi(m)})$ as formal polynomials for all permutations~$\pi$ of $\set{1,2,\ldots,m}$. Because the order of the variables is immaterial, we usually just write a symmetric polynomial as $P(X)$, where $X=\set{x_1, \ldots, x_m}$ is the set of variables, and we write $P(T)$ for its evaluation on a multiset~$T$ of values.

We next recall two important symmetric polynomials and the relationship between them.
For a non-negative integer~$i$, the $i$th \emph{power sum} of a set~$X$ of variables is defined as
\begin{equation}
    \label{eq:power-sum}
    p_i(X) := \sum_{x \in X} x^i \ \text{.}
\end{equation}
(Recall that $0^0:=1$.)
For $1 \leq i \leq \card{X}$, the $i$th \emph{elementary symmetric polynomial} of~$X$ is defined as
\begin{equation}
    \label{eq:elementary-symm-poly}
    e_i(X) := \sum_{\substack{Z \subseteq X, \\ \card{Z}=i}} \; \prod_{z \in Z} z \ \text{.}
\end{equation}
That is, $e_i(X)$ is the multilinear polynomial whose monomials consist of all products of~$i$ distinct variables from~$X$.
We extend this definition to $i = 0$ by setting $e_0(X) := 1$ and to integers $i > \card{X}$ by setting $e_i(X) := 0$.

Power sums and elementary symmetric polynomials are related by \emph{Newton's identities} (see, e.g.,~\cite{Mead1992}), which assert that for $1 \leq i \leq \card{X}$,
\begin{equation}
    \label{eq:newtons-identities}
    i \cdot e_i(X) = \sum_{j = 1}^i (-1)^{j-1} \cdot e_{i-j}(X) \cdot p_j(X) \ \text{.}
\end{equation}

The following standard claim uses Newton's identities to show that if the first~$k$ power sums of two \emph{multisets} of field elements coincide, then so do the first~$k$ elementary symmetric polynomials of those multisets.

\begin{claim}
\label{clm:power-sums-elementary-symmetric}
Let $T,U$ be multisets over a prime field~$\F_q$, let $k \leq q$ be a positive integer, and suppose that $p_i(T) = p_i(U)$ for all $i \in [k]$.
Then $e_i(T) = e_i(U)$ for all $i \in [k]$.
\end{claim}

\begin{proof}
The proof is by (strong) induction.
For the base case where $i = 0$, we have by definition that $e_0(T) = e_0(U) = 1$.
For the inductive case where $1 \leq i < k$, because $i \neq 0 \in \F_q$ we have that
\begin{equation*}
e_i(T) = i^{-1} \cdot \sum_{j = 1}^i (-1)^{j-1} \cdot e_{i-j}(T) \cdot p_j(T) 
            = i^{-1} \cdot \sum_{j = 1}^i (-1)^{j-1} \cdot e_{i-j}(U) \cdot p_j(U) 
            = e_i(U) \ \text{,}
\end{equation*}
where the first and third equalities follow from Newton's identities (\cref{eq:newtons-identities}), and the second equality follows from the claim's hypothesis and the inductive hypothesis (note that the sums involve elementary symmetric polynomials~$e_{i-j}$ only for $i-j < i$).
\end{proof}

We define the \emph{root polynomial} $f_T(x) \in \F[x]$ of a multiset~$T$ over a field~$\F$ to be
\begin{equation} \label{eq:root-polynomial}
f_T(x) := \prod_{t \in T} (x - t) = \sum_{i = 0}^{\card{T}} (-1)^i \cdot e_i(T) \cdot x^{\card{T}-i} \ \text{.}
\end{equation}

We then get the following result, which uses \cref{clm:power-sums-elementary-symmetric} to show that if sufficiently many of the initial power sums of two multisets are equal, then the multisets themselves are equal.

\begin{proposition}
\label{prop:power-sums-root-polynomials}
Let~$q$ be a prime, let $k \leq q/2$ be a positive integer, let $T,U$ be multisets over~$\F=\F_q$ of total cardinality $\card{T}+\card{U} < 2k$, and suppose that $p_i(T) = p_i(U)$ for all $i \in [k]$. Then $T = U$.
\end{proposition}

\begin{proof}
Because \[ \card{T} \equiv p_0(T) = p_0(U) \equiv \card{U} \pmod{q} \] and $0 \leq \card{T} + \card{U} < 2k \leq q$, it follows that $\card{T} = \card{U}$ and hence both $f_T(x), f_U(x)$ have degree $\card{T} < k$.

Next, by the hypotheses and \cref{clm:power-sums-elementary-symmetric}, we have that $e_i(T) = e_i(U)$ for all $i \leq \card{T}$.
Therefore, by the equality in \cref{eq:root-polynomial}, $f_T(x)$ and $f_U(x)$ are identical as formal polynomials in $\F[x]$.
Finally, because the polynomial ring $\F[x]$ is a unique factorization domain, and because $f_T(x)$ and $f_U(x)$ split over~$\F$ by construction, it follows that $T = U$.
\end{proof}

\subsection{Locally Dense Lattices}
\label{sec:locally-dense}

Roughly speaking, \emph{locally dense lattices} are lattices that have one or more cosets with many relatively short vectors.
Somewhat more precisely, a locally dense lattice consists of an integer lattice $\lat \subset \Z^{n}$ and a shift $\vec{x} \in \Z^{n}$ such that for some $\alpha \in (0, 1)$, the number of points in the coset $\vec{x}+\lat$ of norm at most $\alpha \cdot \lambda_1(\lat)$ is large (for our purposes, greater than $2^{n^{\eps}}$ for some constant $\eps > 0$).
Therefore, locally dense lattices are not efficiently list decodable, even combinatorially, to within distance $\alpha \cdot \lambda_1(\lat)$ in the worst case (in particular, around center~$-\vec{x}$).
For the purposes of proving hardness, we also require a linear map~$T$ that projects the short vectors in $\vec{x}+\lat$ onto a lower-dimensional hypercube $\bit^r$.

\begin{definition}
\label{def:locally-dense}
For $p \in [1, \infty)$, real $\alpha > 0$, and positive integers~$r$ and~$R$, a $(p, \alpha, r, R)$-\emph{locally dense lattice} consists of an
integer lattice of rank~$R$ (and some dimension~$n$) represented by a basis matrix $A \in \Z^{n \times R}$,
a positive integer~$\ell$,
a shift $\vec{x} \in \Z^{n}$,
and a matrix $T \in \Z^{r \times n}$, where
\begin{enumerate}[itemsep=0pt]
    \item \label{item:min-dist} $\lambda_1^{(p)}(\lat(A)) \geq \ell^{1/p}$ and
    \item \label{item:cover-hypercube} $\bit^r \subseteq T(V) := \set{T \vec{v} : \vec{v} \in V}$, where $V := (\vec{x} + \lat(A)) \cap \B_p^{n}(\alpha \cdot \ell^{1/p})$ is the set of all vectors of~$\ell_p$ norm at most~$\alpha \cdot \ell^{1/p}$ in the lattice coset $\vec{x}+\lat(A)$.
\end{enumerate}
\end{definition}

A useful tool for satisfying \cref{item:cover-hypercube} in the above definition is the following probabilistic version of Sauer's Lemma due to Micciancio~\cite{journals/siamcomp/Micciancio00}. It roughly says that for $n \gg r$, for any large enough collection of vectors~$W \subseteq B_{n,h}$ (the weight-$h$ slice of $\bit^{n}$), and for a random matrix $T \in \bit^{r \times n}$ whose coordinates are sampled independently with a suitable bias, $\bit^r \subseteq T(W)$ with good probability. We emphasize that all the arithmetic in this theorem is done over the integers (not over $\F_2$).

\begin{theorem}[{\cite[Theorem~4]{journals/siamcomp/Micciancio00}}]
\label{thm:prob-sauer}
Let~$r, n, h$ be positive integers, let $W \subseteq B_{n,h}$, and let $\eps > 0$. If $\card{W} \geq h! \cdot n^{24 r \sqrt{h}/\eps}$ and $T \in \bit^{r \times n}$ is sampled by setting each entry to~$1$ independently with probability $1/(4hr)$, then $\bit^r \subseteq T(W)$ with probability at least $1 - \eps$.
\end{theorem}

\input{reduction}

%% file: reduction.tex
\subsection{Hardness of \texorpdfstring{$\GapSVP$}{GapSVP} via Locally Dense Lattices}

We next recall a variant of (the decision version of) the Closest Vector Problem~(CVP), which will be the hard problem that we reduce to $\GapSVP_p$.
In this variant, called $\GapCVP'_p$, the target vector is either within a specified distance of a lattice vector given by a \emph{binary} combination of basis vectors, or \emph{all non-zero integer multiples} of the target vector are more than a~$\gamma$ multiple of this distance from the lattice (where distance is measured in the $\ell_p$ norm).

\begin{definition}
\label{def:gapcvp-prime}
For $p \in [1, \infty]$, an instance of the $\gamma$-$\GapCVP'_p$ problem consists of a rank-$r$ lattice basis $B \in \Z^{d \times r}$, a target vector~$\vec{t} \in \Z^d$, and a distance threshold $s > 0$.
The goal is to determine whether an input is a YES instance or a NO instance, where these are defined as follows:
\begin{itemize}
    \item YES instance: there exists a \emph{binary} $\vec{c} \in \bit^r$ such that $\norm{B\vec{c} - \vec{t}}_p \leq s$. 
    \item NO instance: $\dist_p(w \vec{t}, \lat(B)) > \gamma s$ for all $w \in \Z \setminus \set{0}$.
\end{itemize}
\end{definition}

\noindent The following hardness theorem follows via a reduction from Exact Set Cover to $\GapCVP'_p$.

\begin{theorem}[\cite{journals/jcss/AroraBSS97}]
\label{thm:cvp-hardness}
For every $p \in [1, \infty)$ and every constant $\gamma \geq 1$, $\gamma$-$\GapCVP'_p$ is $\NP$-hard.
\end{theorem}

The following theorem gives a polynomial-time reduction from $\gamma$-$\GapCVP'_p$ to $\gamma'$-$\GapSVP_p$ for some approximation factors $\gamma > \gamma' \geq 1$, which uses a locally dense lattice as advice.
In general, this advice makes the reduction non-uniform, but when the advice is efficiently computable by a (randomized) algorithm, as it is in this and prior works, the procedure is an efficient (randomized) reduction.
The reduction below is very similar to the one in~\cite[Theorem~5.1]{journals/toc/Micciancio12}, but written so as to allow for using an arbitrary locally dense lattice as advice.
Due to this similarity, and for concision, we defer its proof to \cref{sec:cvp-svp-reduction-proof}.

\begin{restatable}{theorem}{cvpsvp}
\label{thm:cvp-to-svp}
Let $p \geq 1$, $r$ and~$n$ be positive integers, $\alpha > 0$ be a constant, and $\gamma, \gamma'$ be constants satisfying
\[
1/\alpha > \gamma' \geq 1 \text{ and } \gamma \geq \gamma' \cdot \Big(\frac{1}{1 - (\alpha \gamma')^p}\Big)^{1/p} \ \text{.}
\]
There is a deterministic polynomial-time algorithm that, given a $\gamma$-$\GapCVP'_p$ instance $(B, \vec{t}, s)$ of rank~$r$ and a $(p, \alpha, r, R)$-locally dense lattice $(A, \ell, \vec{x}, T)$ as input, outputs a $\gamma'$-$\GapSVP_p$ instance $(B', s')$ of rank $R+1$ which is a YES (respectively, NO) instance if $(B, \vec{t}, s)$ is a YES (resp., NO) instance.
\end{restatable}

\noindent From these two theorems we get the following hardness results for $\GapSVP$.

\begin{corollary}
\label{cor:svp-hardness}
Let $p \geq 1$, let $r$ be a positive integer, let $\alpha > 0$ be a constant, and suppose that there is an algorithm $A$ that computes a $(p, \alpha, r, \poly(r))$-locally dense lattice in $\poly(r)$ time.
Let $\gamma$ be a constant satisfying $1 \leq \gamma < 1/\alpha$. Then:
\begin{enumerate}
    \item \label{item:det-svp-hardness} If~$A$ is deterministic, then $\gamma$-$\GapSVP_p$ is $\NP$-hard (and exact $\GapSVP_p$ is $\NP$-complete).
    \item \label{item:one-sided-error-svp-hardness} If $A$ is randomized and its output satisfies \cref{item:min-dist} of \cref{def:locally-dense} with probability~$1$ and \cref{item:cover-hypercube} of \cref{def:locally-dense} with probability at least $2/3$, then $\gamma$-$\GapSVP_p$ is not in $\RP$ unless $\NP \subseteq \RP$.\footnote{The condition ``$\gamma$-$\GapSVP_p$ is not in $\RP$'' is a slight abuse of notation, since $\gamma$-$\GapSVP_p$ for $\gamma > 1$ is a promise problem rather than a language.
    However, the definition of $\RP$ can naturally be extended to encompass promise problems, which is the intended meaning here.}
    \item \label{item:two-sided-error-svp-hardness} If~$A$ is randomized, and its output satisfies \cref{item:min-dist,item:cover-hypercube} of \cref{def:locally-dense} with probability at least $2/3$, then there is no \emph{randomized} polynomial-time algorithm for $\gamma$-$\GapSVP_p$ unless $\NP \subseteq \BPP$.
\end{enumerate}
\end{corollary}

\begin{proof}
\cref{item:det-svp-hardness,item:two-sided-error-svp-hardness} follow immediately by combining \cref{thm:cvp-hardness,thm:cvp-to-svp}.
Inspection of the proof of \cref{thm:cvp-to-svp} shows that for NO instances to be mapped to NO instances, only \cref{item:min-dist} of \cref{def:locally-dense} is needed, from which \cref{item:one-sided-error-svp-hardness} of the claim follows.
\end{proof}

%% file: local-density.tex
\section{Local Density from Reed-Solomon Codes}
\label{sec:ldl-reed-solomon}

In this section we show how to obtain locally dense lattices from Reed-Solomon codes with appropriate parameters.
More specifically, we show to satisfy \cref{def:locally-dense} using a lattice $\lat := \latpar(H)$ corresponding to a parity-check matrix $H = H_q(k, S)$ from \cref{def:parity-check-matrix}. (Recall that $H$ is the parity-check matrix of a Reed-Solomon code when $S = \F_q$, and of a generalized Reed-Solomon code for any $S \subseteq \F_q$.)

The overall structure of the argument is as follows.
First, in \cref{sec:rs-min-dist} we give a lower bound of $\lambda_1^{(p)}(\lat) \geq (2k)^{1/p}$, which corresponds to \cref{item:min-dist} of \cref{def:locally-dense}, by using the connection between power sums and symmetric polynomials (see \cref{sec:symmetric-polynomials}).
Then, in \cref{sec:rs-dense-cosets} we use the upper bound $\det(\lat) \leq q^k$ from \cref{lem:parity-lat-properties} and the pigeonhole principle to show that there exists a lattice coset with many short (binary) vectors, and in fact a suitably sampled random coset has this property with good probability.
Finally, in \cref{sec:rs-main-argument} we set parameters and use \cref{thm:prob-sauer} to satisfy \cref{item:cover-hypercube} of \cref{def:locally-dense} with good probability.

\subsection{Minimum Distance}
\label{sec:rs-min-dist}

The following theorem says that for any $k \leq \card{S}/2$, the~$\ell_1$ minimum distance of $\lat=\latpar(H)$ for $H=H_q(k,S)$ is at least~$2k$.
The theorem and proof are very similar to one that lower bounds the minimum distance of Craig lattices, as given in~\cite{Craig} and~\cite[Chapter 8, Theorem 7]{conway1999sphere}.

Note that a weaker bound of $\lambda_1^{(1)}(\lat) \geq k+1$ (for any $k < q$) follows trivially from the minimum Hamming distance $k+1$ of the (generalized Reed-Solomon) code having parity-check matrix~$H$.
However, this bound is not strong enough for the rest of the local-density argument below, which requires $\lambda_1^{(1)}(\lat) \geq (1+\Omega(1)) k$.

\begin{theorem}
\label{thm:rs-min-dist}
Let~$q$ be a prime, let $S \subseteq \F_q$, let $k \leq \card{S}/2$ be a positive integer, and let $H := H_{q}(k,S) \in \F_q^{k \times S}$ be the matrix from \cref{def:parity-check-matrix}.
Then $\lat = \latpar(H)$ has~$\ell_1$ minimum distance $\lambda^{(1)}_1(\lat) \geq 2k$.

\noindent As a consequence, for any $p \in [1,\infty)$ the~$\ell_p$ minimum distance satisfies $\lop(\lat) \geq (2k)^{1/p}$.
\end{theorem}

We point out that the $2^{1/p}$ factor in \cref{thm:rs-min-dist} propagates to the relative-distance bound for local density in \cref{thm:ldl-from-rs} below, and then to the $\GapSVP$ approximation factor in our main hardness theorem, \cref{thm:svp-hardness}.

\begin{proof}
The consequence follows immediately from the fact that $\lat \subseteq \Z^S$ and $\norm{\vec{v}}_p \geq \norm{\vec{v}}_1^{1/p}$ for all $\vec{v} \in \Z^S$.

Now consider some arbitrary $\vec{x} \in \lat \subseteq \Z^{S}$ for which $\norm{\vec{x}}_1 < 2k$; we will show that $\vec{x}=\vec{0}$.
Let $\vec{x}^+, \vec{x}^- \in \Z^{S}$ be the unique non-negative integer vectors satisfying $\vec{x} = \vec{x}^+ - \vec{x}^-$.
Define multisets $T^+$ and $T^-$ over~$S$ that respectively depend on~$\vec{x}^+$ and $\vec{x}^-$ as follows.
For each $s \in S$ with $x_s^+ > 0$ (respectively, $x_s^- > 0$), let~$T^+$ (respectively, $T^-$) contain~$s$ with multiplicity~$x_s^+$ (respectively, $x_s^-$).\footnote{For example, if $S = \set{0,1,2,3,4} = \F_q$ and $\vec{x} = (x_0, x_1, x_2, x_3, x_4)^t = (1, -2, 0, 1, 0)^t$, then $x^+ = (1, 0, 0, 1, 0)$, $x^- = (0, 2, 0, 0, 0)$, and accordingly $T^+ = \set{0, 3}$, $T^- = \set{1, 1}$.}

Note that $\card{T^+} + \card{T^-} = \norm{\vec{x}}_1 < 2k$.
Because $H\vec{x} = H(\vec{x}^+ - \vec{x}^-) = \vec{0} \in \F_q^{k}$, by definition of~$H$ we have that $p_i(T^+) = p_i(T^-)$ for all $i \in [k]$ (where recall that~$p_i$ denotes the $i$th power sum).
Because $k \leq \card{S}/2 \leq q/2$, by \cref{prop:power-sums-root-polynomials} it follows that $T^+ = T^-$. Since $T^+ \cap T^- = \emptyset$ by construction, we must have $T^+ = T^- = \emptyset$, and hence $\vec{x}=\vec{0}$, as desired.
\end{proof}

The following lemma (which is well known in other forms) shows that the lower bound $\lambda_1^{(p)}(\latpar(H)) \geq (2k)^{1/p}$ from \cref{thm:rs-min-dist} is in fact an equality under mild conditions on the parameters, by giving an explicit lattice coset that has multiple short vectors.
However, because it proves only that the number of such vectors is polynomial in the dimension, it is insufficient to establish local density.

\begin{lemma}
\label{lem:rs-min-dist-upper}
Let~$q$ be a prime, let $k$ be a positive integer that divides~$q-1$, and let $H := H_q(k,S) \in \F_q^{k \times S}$ where $\F_q^* \subseteq S \subseteq \F_q$.
Then for $\vec{u} := (k, 0, \ldots, 0) \in \F_q^{k}$, the lattice coset $\latpar_{\vec{u}}(H) = \set{\vec{x} \in \Z^S : H \vec{x} = \vec{u}}$ contains $(q-1)/k$ binary vectors of Hamming weight~$k$ and pairwise disjoint support.
As a consequence, when~$k < q-1$, we have $\lambda_1^{(p)}(\latpar(H)) = (2k)^{1/p}$ for any $p \in [1,\infty)$.
\end{lemma}

\begin{proof}
Let~$G$ be the order-$k$ subgroup of the (cyclic, multiplicative) group $\F_q^*$, i.e., the subgroup of the $k$th roots of unity.
Then the binary indicator vectors $\vec{x}_{C} \in \bit^S$ of each of the $(q-1)/k$ pairwise disjoint cosets $C = cG$ of~$G$ all belong to the coset $\latpar_{\vec{u}}(H)$.
This is simply because for any such coset, the $0$th power sum of its elements is~$k$, and the $i$th power sum for $0 < i < k$ is zero;
this can be seen by Newton's identities and the fact that the root polynomial of~$C$ is $f_C(x) = \prod_{c \in C} (x-c) = x^k - r_C$, where $r_C = c^k$ for every $c \in C$.
Finally, when~$k < q-1$, there is more than one such vector~$\vec{x}_C$, and the differences between distinct pairs of them are lattice vectors in $\set{0, \pm 1}^S$ of Hamming weight~$2k$, and hence~$\ell_p$ norm~$(2k)^{1/p}$.
\end{proof}

\subsection{Dense Cosets}
\label{sec:rs-dense-cosets}

Following an approach previously used in~\cite{journals/siamcomp/Micciancio00,journals/jacm/Khot05,journals/toc/Micciancio12} (and implicitly in~\cite{DBLP:conf/stoc/Ajtai98}),
we first show via a pigeonhole argument that a dense lattice coset must exist, and then show how to sample such a coset efficiently (with good probability).

For a prime~$q$, a positive integer~$k$, and a set $S \subseteq \F_q$ of size~$n$ (with some arbitrary ordering of its elements), let $H = H_q(k, S) \in \F_q^{k \times n}$ be the parity-check matrix from \cref{def:parity-check-matrix}.
By \cref{lem:parity-lat-properties}, the lattice $\lat = \latpar(H) \subseteq \Z^n$ has $\det(\lat) \leq q^k$ integer cosets.
Recall that $B_{n,h}$ is the set of $n$-dimensional binary vectors of Hamming weight~$h$, which has cardinality $\card{B_{n,h}} = \binom{n}{h}$.
Therefore, by the pigeonhole principle, there must exist some integer coset $\vec{x} + \lat$ with $\card{(\vec{x} + \lat) \cap B_{n, h}} \geq \binom{n}{h}/q^k$ weight-$h$ binary vectors.
In particular, taking $n \approx q$, $h \approx \alpha^p \cdot (2k)$ for some constant $\alpha > 1/2^{1/p}$, and $k = q^{\eps}$ for a suitable small constant $\eps > 0$ implies the existence of a coset with roughly $q^{(2\alpha^p - 1) k} = q^{\Omega(q^{\eps})}$ such vectors.
These vectors have~$\ell_p$ norm~$h^{1/p} \approx \alpha \cdot (2k)^{1/p}$, whereas by \cref{thm:rs-min-dist} the lattice minimum distance is at least $(2k)^{1/p}$, yielding a local-density relative distance of roughly $\alpha$.

The following lemma extends the above existential result by showing that something very similar holds for a \emph{uniformly random} shift $\vec{x} \in B_{n, h}$: for any $\delta > 0$, the coset $\vec{x}+\lat$ contains at least $\delta \cdot \binom{n}{h}/q^k$ weight-$h$ binary vectors with probability greater than $1 - \delta$. The proof given below closely follows the structure of the very similar one of~\cite[Lemma~4.3]{journals/jacm/Khot05}.

\begin{lemma}
\label{lem:coset-sampling}
For a prime~$q$, positive integer~$k$, and set $S \subseteq \F_q$ of size~$n$, let $H = H_q(k, S) \in \F_q^{k \times n}$ be the parity-check matrix from \cref{def:parity-check-matrix}.
There is an efficient randomized algorithm that, for any $\delta \geq 0$, and on input~$H$ and any $h \in [n]$, outputs a shift $\vec{x} \in B_{n,h}$ such that \[  \Pr_{\vec{x}}\bracks[\Big]{\card{(\vec{x}+\lat) \cap B_{n, h}} \geq \delta \cdot \binom{n}{h}/q^k} > 1-\delta \ \text{.} \]
\end{lemma}

\begin{proof}
The algorithm simply samples and outputs a uniformly random binary vector $\vec{x} \in B_{n, h}$. This is clearly efficient.
To show correctness, we will use the syndromes of~$H$.
For each $\vec{u} \in \F_q^{k}$, define $K_{\vec{u}} := \card{\set{\vec{z} \in B_{n, h} : H\vec{z} = \vec{u}}}$, and define $\vec{s} := H\vec{x} \in \F_q^k$ to be the syndrome corresponding to $\vec{x}$.
So, we need to prove that
$K_{\vec{s}} \geq \delta \cdot \binom{n}{h}/q^k$ with probability greater than $1 - \delta$.
Indeed, we have:
\begin{align*}
\Pr_{\vec{x}}\bracks[\Big]{\card{(\vec{x} + \lat) \cap B_{n, h}} < \delta \cdot \binom{n}{h}/q^k}
&= \Pr_{\vec{x}}\bracks[\Big]{K_{\vec{s}} < \delta \cdot \binom{n}{h}/q^k} \\
&= \sum_{\vec{u} \in \F_q^{k} : K_{\vec{u}} < \delta \cdot \binom{n}{h}/q^k} \Pr_{\vec{x}}[H\vec{x} = \vec{u}] \\
&= \sum_{\vec{u} \in \F_q^{k} : K_{\vec{u}} < \delta \cdot \binom{n}{h}/q^k}
\frac{K_{\vec{u}}}{\binom{n}{h}} \\
&< \sum_{\vec{u} \in \F_q^{k} : K_{\vec{u}} < \delta \cdot \binom{n}{h}/q^k} \frac{\delta}{q^k} \\
& \leq \delta \ ,
\end{align*}
where the first inequality uses the fact that the sum is over syndromes $\vec{u}$ with $K_{\vec{u}} < \delta \cdot \binom{n}{h}/q^k$, and the second inequality uses the fact that there are at most $q^k$ terms in the sum.
\end{proof}

\subsection{The Main Argument}
\label{sec:rs-main-argument}

\begin{theorem}[Locally dense lattices from Reed-Solomon codes]
\label{thm:ldl-from-rs}
For any $p \in [1, \infty)$ and constant $\alpha > 1/2^{1/p}$, there exists a randomized polynomial-time algorithm that, given any sufficiently large positive integer~$r$ in unary as input, outputs a $(p, \alpha, r, R=\poly(r))$-locally dense lattice (\cref{def:locally-dense}) with probability at least~$2/3$.
Moreover, the algorithm's output satisfies \cref{item:min-dist} of \cref{def:locally-dense} with probability~$1$.
\end{theorem}

\begin{proof}
The algorithm starts by setting its parameters as follows.
It sets $\eps := 2\alpha^p - 1 > 0$, and chooses:
\begin{itemize}[itemsep=0pt]
    \item a $\poly(r)$-bounded integer $k \geq r^{1/(1/2 - \delta)}$ for some arbitrary constant $\delta \in (0, 1/2)$, and
    \item a $\poly(r)$-bounded prime~$q \geq k^{3 (1 + \eps)/\eps}$. (Such a prime~$q$ always exists by Bertrand's Postulate.)
\end{itemize} 
The algorithm then computes the components of a $(p, \alpha, r, R=q)$-locally dense lattice $(A, \ell, \vec{x}, T)$ as follows.
It lets:
\begin{itemize}[itemsep=0pt]
    \item $A \in \Z^{q \times q}$ be a basis of $\lat := \latpar(H)$, where $H=H(k,S) \in \F_q^{k \times q}$ for $S=\F_q$;\footnote{%
    For appropriate parameters, our argument works more generally for any sufficiently large subset $S \subseteq \F_q$, with $R=\card{S}$; we use $S=\F_q$ for simplicity.}
    \item $\ell := 2k$;
    \item $\vec{x} \in B_{q,h}$ be a uniformly random $q$-dimensional binary vector of Hamming weight~$h := \floor{(1 + \eps) k}$; 
    \item $T \in \bit^{r \times q}$ be chosen by independently setting each of its entries to be~$1$ with probability $1/(4hr)$, and to be~$0$ otherwise.
\end{itemize}
It then outputs $(A, \vec{x}, \ell, T)$.

We first analyze the algorithm's running time. A suitable prime~$q$ can be found in $\poly(r)$ time using, e.g., trial division (recall that~$r$ is given in unary). The basis~$A$ can be computed in deterministic polynomial time from the generating set of column vectors $(B \mid q I_q)$, where~$B$ is a basis of $\ker(H) \subseteq \F_q^q$ (lifted to the integers).
It is clear that~$\ell$ can be computed in deterministic polynomial time, and that~$\vec{x}$ and~$T$ can be computed in randomized polynomial time. So, the algorithm runs in randomized polynomial time.

It remains to show correctness, i.e., that $(A, \vec{x}, \ell, T)$ satisfies the two conditions in \cref{def:locally-dense} with suitable probability over the random choices of~$\vec{x}$ and~$T$.
First, \cref{item:min-dist} is always satisfied, because by \cref{thm:rs-min-dist} we have
\[ \lambda_1(\lat) \geq (2k)^{1/p} = \ell^{1/p} \ \text{.} \]

In the rest of the proof we consider \cref{item:cover-hypercube} of \cref{def:locally-dense}. Let $W := (\vec{x} + \lat) \cap B_{q,h}$.
Because
\[ \norm{\vec{w}}_p^p = h \leq (1+\eps)k = \alpha^p \cdot \ell \]
for each $\vec{w} \in W$, we have $W \subseteq V := (\vec{x}+\lat) \cap \B_p^{q}(\alpha \cdot \ell^{1/p})$.

By \cref{lem:coset-sampling},
$\Pr_{\vec{x}}[\card{W} \geq \binom{q}{h}/(10 q^k)] > 1 - 1/10 = 9/10$.
If this event holds, and
\begin{equation}
\label{eq:parameters-apply-sauer}
\frac{\binom{q}{h}}{10 q^k} \geq h! \cdot q^{240 r \sqrt{h}} \ \text{,}
\end{equation}
then by \cref{thm:prob-sauer} we have $\bit^r \subseteq T(W) \subseteq T(V)$ with probability at least $1 - 1/10 = 9/10$ (over the choice of~$T$).
So, it suffices to show that the condition in \cref{eq:parameters-apply-sauer} holds for all sufficiently large~$k$, and hence for all sufficiently large~$r$.
By taking a union bound over the $1/10$ failure probabilities from \cref{lem:coset-sampling} and \cref{thm:prob-sauer}, we get that the algorithm's overall success probability is at least $1 - 2/10 > 2/3$ for all sufficiently large~$r$, as needed.

Using the standard bound $\binom{q}{h} \geq (q/h)^h$ for binomial coefficients and that $h \geq (1 + \eps)k - 1$, we have that
\begin{equation} \label{eq:parameters-setting-lb}
\frac{\binom{q}{h}}{10 q^k} \geq \frac{q^{h-k}}{10 h^h}
= \Omega \Big(\frac{q^{\eps k - 1}}{h^h} \Big) \ \text{.}
\end{equation}
Furthermore, by the choice of~$k$ relative to~$r$ and $h \leq (1+\eps)k$, we have that
\begin{equation} \label{eq:parameters-setting-ub}
h! \cdot q^{240 r \sqrt{h}} \leq h^h \cdot q^{240 k^{1/2 - \delta} \sqrt{(1 + \eps)k}} \leq h^h \cdot q^{o(k)} \ \text{.}
\end{equation}
So, by combining \cref{eq:parameters-setting-lb,eq:parameters-setting-ub}, in order to establish \cref{eq:parameters-apply-sauer} it suffices to show that $q^{(1-o(1)) \eps k} \geq h^{2h}$. By taking logs, this is equivalent to
\begin{equation}
\label{eq:parameters-setting-after-logs}
(1-o(1)) \cdot \eps k \log q \geq 2h \log h \ \text{.}
\end{equation}
Finally, using that $k^{3(1+\eps)/\eps} \leq q \leq \poly(k)$ and $h \leq (1 + \eps)k$, in order for \cref{eq:parameters-setting-after-logs} to hold it suffices to have 
\[
(1-o(1)) \cdot \eps k \cdot \frac{3 (1 + \eps)}{\eps} \cdot \log k
= (3 - o(1)) \cdot (1 + \eps) \cdot k \log k
\geq 2 (1 + \eps) \cdot k \log k + O(k) \ \text{,}
\]
which indeed holds for all sufficiently large~$k$, as needed.
\end{proof}

We emphasize that \cref{thm:ldl-from-rs} uses randomness only to sample~$\vec{x}$ and~$T$.
As an immediate corollary, we obtain our main hardness result, \cref{thm:svp-hardness}---which, to recall, asserts that for all constants $p \in [1, \infty)$ and $\gamma < 2^{1/p}$, there is no polynomial-time algorithm for $\gamma$-$\GapSVP_p$ unless $\NP \subseteq \RP$.

\begin{proof}[Proof of \cref{thm:svp-hardness}]
Combine \cref{item:one-sided-error-svp-hardness} of \cref{cor:svp-hardness} with \cref{thm:ldl-from-rs}.
\end{proof}

%% file: decoding.tex
\section{Efficient Decoding Near Minkowski's Bound}
\label{sec:decoding}

In this section, we show that a recent result of Mook and Peikert~\cite{DBLP:journals/tit/MookP22}, which builds on work of Guruswami and Sudan~\cite{DBLP:journals/tit/GuruswamiS99} and Koetter and Vardy~\cite{DBLP:journals/tit/KoetterV03a} on list-decoding Reed-Solomon codes, yields a polynomial-time algorithm for decoding lattices $\lat = \latpar(H)$ with $H = H_q(k, \F_q)$ up to distance $\Theta(\sqrt{k})$.
We additionally observe that by choosing $k = \Theta(q/\log q)$, such lattices are asymptotically nearly tight with Minkowski's bound (\cref{eq:minkowski}).
Putting these observations together, we obtain an efficient algorithm for decoding to a distance within a $O(\sqrt{\log q})$ factor of Minkowski's bound (here $q = n$ is the lattice rank and dimension).

\subsection{Construction and Algorithm}
\label{subsec:minowski-alg}

Define the additive quotient group $\R_q := \R/(q \Z)$ and the Euclidean norm of any $\hat{\vec{y}} \in \R_q^n$ as
\begin{equation}
\label{eq:Rq-norm}
\norm{\hat{\vec{y}}} := \min \set{\norm{\vec{y}} : \vec{y} \in \hat{\vec{y}} + q \Z^n} \ \text{.}
\end{equation}
Equivalently, $\norm{\hat{\vec{y}}}$ is the standard $\R^n$ Euclidean norm of the unique real vector $\vec{y} \equiv \hat{\vec{y}} \pmod{q \Z^n}$ having coordinates in $[-q/2, q/2)$.
In additive arithmetic that mixes elements of~$\F_q$ and~$\R_q$, we implicitly `lift' the former to the latter in the natural way.

We again use the fact that for evaluation set $S = \F_q$, the matrix $H = H_q(k, \F_q)$ defined in \cref{eq:parity-check} is a parity-check matrix of the Reed-Solomon code $\RS_q[q - k, \F_q]$, and therefore $\latpar(H_q(k, \F_q)) = \RS_{q}[q - k, \F_q] + q \Z^q$.
This view lets us take advantage of the decoding algorithm from the following theorem of~\cite{DBLP:journals/tit/MookP22}, which gives an efficient (list) decoder in the~$\ell_2$ norm for Reed-Solomon codes.\footnote{\label{foot:decoding-notes}In fact, the cited result from \cite{DBLP:journals/tit/MookP22} is more general, giving a decoder for \emph{$\F_p$-subfield subcodes} of Reed-Solomon codes over finite fields of order $q=p^r$, for a prime~$p$.
Here we need only the special case where the Reed-Solomon code is over a prime field (i.e., where $r = 1$).
On the other hand, we note that if \cref{prop:efficient-rs-decoding} were extended to handle \emph{generalized} Reed-Solomon codes, then we would get a corresponding strengthening of \cref{cor:efficient-consARS-decoding} for decoding lattices $\latpar(H_q(k,S))$ with general $S$, not just $S = \F_q$.}

\begin{proposition}[{\cite[Algorithm~1 and Theorem~3.4]{DBLP:journals/tit/MookP22}}] \label{prop:efficient-rs-decoding}
Let~$q$ be a prime, $S \subseteq \F_q$ be an evaluation set of size $n := \card{S}$, $k \leq n$ be a nonnegative integer, and $\eps > 0$.
There is a deterministic algorithm that, on input~$q$,~$S$,~$k$,~$\eps$,
and a vector $\hat{\vec{y}} \in \R_q^n$, outputs all codewords $\vec{c} \in \RS_{q}[n - k, S]$ such that $\norm{\hat{\vec{y}} - \vec{c}}^2 \leq (1 - \eps) (k + 1)/2$, in time polynomial in $n$, $\log q$, and $1/\eps$.\footnote{\label{foot:poly-received-word}Formally, the runtimes of the decoding algorithms in \cref{prop:efficient-rs-decoding,cor:efficient-consARS-decoding} additionally depend on the lengths of the respective ``received words'' $\hat{\vec{y}}$ and $\vec{y}$ that they take as input, which must be specified to finite precision. However, for simplicity we describe the algorithms in the ``Real RAM model,'' while noting that their runtime dependence on the encoding lengths of $\hat{\vec{y}}, \vec{y}$ is polynomial.} 
\end{proposition}

The following corollary, which is the main result of this section, says that by taking $S = \F_q$ and $k = \Theta(q/\log q)$, (1)~the root Hermite factor of $\latpar(H)$ is within an $O(\sqrt{\log q})$ factor of Minkowski's bound (\cref{eq:minkowski}), and (2)~it is possible to efficiently decode this lattice to a distance of $\Omega(\sqrt{k}) = \Omega(\sqrt{q/\log q})$, which is again within an $O(\sqrt{\log q})$ factor of Minkowski's bound.

We remark that by setting $\eps \leq 1/(k+1)$ in \cref{cor:efficient-consARS-decoding}, we get efficient decoding to a distance at least $\sqrt{k/2}$ but less than $\sqrt{(k+1)/2}$, which is slightly more than half the lower bound of $\sqrt{2k}$ on the minimum Euclidean distance of the lattice (\cref{thm:rs-min-dist}). Recall that this lower bound is tight when~$k$ is a proper divisor of~$q-1$ (see \cref{lem:rs-min-dist-upper}), so with this parameterization we get efficient \emph{list} decoding (i.e., the algorithm may return more than one lattice vector) slightly beyond the unique-decoding bound of half the minimum distance.

\begin{corollary}[Efficient decoding near Minkowski's bound] \label{cor:efficient-consARS-decoding}
Let $H = H_q(k, \F_q)$ for a prime $q$ and $k := \floor{q/(2 \log q)} \leq q/2$, where all logarithms are base two. Then for $\lat := \latpar(H) \subseteq \Z^q$:
\begin{enumerate}
    \item \label{item:high-density} $\sqrt{q/\log q - 2} \leq \sqrt{2 k} \leq \lambda_1(\lat) \leq \sqrt{q} \cdot \det(\lat)^{1/q} \leq \sqrt{2 q}$.

    \item \label{item:efficient-decoding} For any $\eps > 1/\poly(q)$, there is an algorithm that, on input~$q$ and a vector $\vec{y} \in \R^q$, outputs all lattice vectors $\vec{v} \in \lat$ satisfying $\norm{\vec{y} - \vec{v}} \leq \sqrt{(1-\eps)(k+1)/2}$ in time polynomial in~$q$.
\end{enumerate}
\end{corollary}

\begin{proof}
For \cref{item:high-density}, we have
\[
\sqrt{q/\log q - 2} \leq \sqrt{2k} \leq \lambda_1(\lat) \leq \sqrt{q} \cdot \det(\lat)^{1/q} = \sqrt{q} \cdot q^{k/q} \leq \sqrt{2 q} \ \text{.}
\]
The first inequality follows from the choice of $k$, the second inequality is by \cref{thm:rs-min-dist}, the third inequality is Minkowski's bound (\cref{eq:minkowski}), the equality follows from \cref{lem:parity-lat-properties} (recall that the rows of $H$ are linearly independent), and the final inequality again follows from the choice of $k$.\footnote{Analyzing the derivative of $\log(\sqrt{2k}/q^{k/q})$
with respect to~$k$ shows that our choice of~$k$ is asymptotically optimal for maximizing the root Hermite factor of $\latpar(H_q(k, \F_q))$.}

The algorithm claimed in \cref{item:efficient-decoding} works as follows.
First, it computes $k$ and
$\hat{\vec{y}} = \vec{y} \bmod q\Z^q \in \R_q^q$ from the input $q$ and $\vec{y}$.
It then calls the algorithm from \cref{prop:efficient-rs-decoding} on $q$, $S = \F_q$, $k$, $\eps$, and $\hat{\vec{y}}$, and receives as output zero or more codewords $\vec{c} \in \RS_{q}[q - k, \F_q]$.
For each such~$\vec{c}$, it outputs the unique vector
$\vec{v} := \argmin_{\vec{v}' \in \vec{c} + q\Z^q} \norm{\vec{y} - \vec{v}'} \in \lat$.

The value $k$ and
vectors $\hat{\vec{y}}$, $\vec{v}$ can be computed efficiently (assuming that~$\vec{v}$ is well defined), so it is clear from \cref{prop:efficient-rs-decoding} that this algorithm runs in time polynomial in $q$ (recall that the dimension $n = q$).
It remains to show correctness.
First, it is immediate from the definitions that for any $r < q/2$, the function $f(\vec{v}) = \vec{v} \bmod q \Z^q$ is a bijection from the set of lattice vectors
\[
\set{\vec{v} \in \lat : \norm{\vec{y} - \vec{v}} \leq r} \ \text{,}
\]
to the set of codewords
\[
\set{\vec{c} \in \RS_{q}[q - k, \F_q] : \norm{\hat{\vec{y}} - \vec{c}} \leq r} \ \text{,}
\]
and that $g(\vec{c}) := \argmin_{\vec{v}' \in \vec{c} + q\Z^q} \norm{\vec{y} - \vec{v}'}$ is the inverse function of~$f$, i.e., $g=f^{-1}$.
Moreover, because $q \geq 2$, we have that the decoding distance~$r$ satisfies
\[
r := \sqrt{(1 - \eps)(k + 1)/2} \leq \sqrt{(1 - \eps) (q/(2 \log q) + 1)/2} \leq \sqrt{1 - \eps} \cdot q/2 < q/2 \ \text{.}
\]
Because the algorithm from \cref{prop:efficient-rs-decoding} outputs (exactly) $\set{\vec{c} \in \RS_{q}[q - k, \F_q] : \norm{\hat{\vec{y}} - \vec{c}} \leq r}$, it follows that the algorithm described above outputs (exactly) $\set{\vec{v} \in \lat : \norm{\vec{y} - \vec{v}} \leq r}$, as needed.
\end{proof}

\begin{remark}
We remark that the main consequence of \cref{item:high-density} of \cref{cor:efficient-consARS-decoding}---namely, an explicit construction of a family of lattices having root Hermite factors within a $O(\sqrt{\log n})$ factor of Minkowski's bound, obtained via Construction~A (where $n$ is the lattice dimension)---only needs a family of codes satisfying milder conditions than what (generalized) Reed-Solomon codes satisfy.
Namely, achieving this result only requires a family of linear $q$-ary codes~$\mathcal{C}$ for prime~$q$ with block length~$n$, codimension $k = \Theta(n/\log n)$, and minimum distance (in the Hamming metric) $d = \Omega(k)$.
The latter is a weaker condition than \emph{maximum distance separability}~(MDS), which requires that $d=k+1$.
Indeed, $d = \Omega(k)$ implies that the corresponding Construction-A lattice $\mathcal{C} + q \Z^n$ has an~$\ell_2$ minimum distance of $\Omega(\min \set{\sqrt{k},q})$, which is $\Omega(\sqrt{k})$ when $k=O(q^2)$.
So, unlike our main hardness result, \cref{cor:efficient-consARS-decoding} does not use
\cref{thm:rs-min-dist}
in any essential way.

Finally, we also note that obtaining a direct analog of \cref{item:efficient-decoding} of \cref{cor:efficient-consARS-decoding}---i.e., efficiently decoding to within an $O(\sqrt{\log n})$ factor of Minkowski's bound on $\mathcal{C} + q \Z^n$---additionally requires an efficient algorithm for decoding~$\mathcal{C}$ to an~$\ell_2$ distance of $\Omega(\sqrt{k})$, but that this is in turn a weaker requirement than what \cref{prop:efficient-rs-decoding} fulfills.
\end{remark}

%% file: derand.tex
\section{Attempted Derandomization}
\label{sec:derand}

In this section we describe some of our unsuccessful attempts, and associated barriers, to adapt the \emph{randomized} locally dense lattice construction from \cref{sec:ldl-reed-solomon} into a \emph{deterministic} one.
More specifically: for some $p \in [1,\infty)$ and $\alpha \in (0,1)$, our goal is a deterministic algorithm that, given any positive integer~$r$ (in unary), outputs a $(p,\alpha,r,R=\poly(r))$-locally dense lattice in $\poly(r)$ time.
Recall from \cref{def:locally-dense} that this requires constructing an explicit rank-$R$ lattice and ``dense coset'' (both integral) containing at least~$2^r$ vectors that have norm at most~$\alpha$ times (a known lower bound on) the lattice's minimum distance.\footnote{\label{foot:dense-coset-not-enough}Recall that a locally dense lattice must also come with a suitable linear transform~$T$.
In this section we do not address the construction of such~$T$, because deterministically constructing a lattice and dense coset is already challenging enough.
We hope that such a construction and its analysis will naturally reveal a suitable choice of~$T$ as well.}

As in \cref{sec:ldl-reed-solomon}, we focus on the case where the lattice has the form $\lat = \latpar(H) \subseteq \Z^S$, where $H=H_q(k,S)$ for some prime~$q$, positive integer~$k$, and set $S \subseteq \F_q$ of cardinality $n=R=\card{S}$ (which all may be set by the algorithm).
Recall from \cref{sec:rs-dense-cosets} that for appropriate parameters, there exists a dense coset $\vec{x}+\lat$ having $n^{\Omega(k)}$ short vectors (indeed, binary vectors of a particular Hamming weight), and that an appropriate random choice of coset has this property with good probability.
To derandomize, we need to deterministically construct a dense integral coset $\vec{x}+\lat$, or equivalently, its syndrome $\vec{u} = H \vec{x} \in \F_q^k$.
This is the problem that we address in this section.
We thank Swastik Kopparty~\cite{kopparty_personal20} for explaining its connections to the approaches considered in \cref{sec:binary-rs-list,sec:fourier}.

As a first observation, we note that \emph{not every} syndrome (coset) is suitable, and some unsuitable ones are easy to describe.
Clearly, the all-zeros syndrome~$\vec{0} \in \F_q^k$ is unsuitable, because it corresponds to (the zero coset of) the lattice~$\lat$ itself, which has only one vector shorter than the minimum distance.
More generally, by the triangle inequality, for any $\vec{x} \in \Z^S$ such that $\norm{\vec{x}} < (1 - \alpha) \lambda_1(\lat)$, the lattice coset $\vec{x} + \lat$ has at most one vector of length at most $\alpha \lambda_1(\lat)$, namely $\vec{x}$ itself (if it is short enough).
So, when $\lambda_1(\lat) \geq \sqrt{2k}$ (e.g., as implied by \cref{thm:rs-min-dist}), if $\norm{\vec{x}} < (1 - \alpha) \sqrt{2k}$ then the coset $\vec{x}+\lat$ has at most one vector of length at most $\alpha \sqrt{2k}$.
(A similar argument holds for all other~$\ell_p$ norms.)
Stated in terms of syndromes: no ``too small'' integer linear combination of the columns of~$H$ is a suitable syndrome, because it corresponds to shifting the lattice by ``too little.''

\subsection{Binary Coset Vectors and Reed-Solomon List Decoding Bounds}
\label{sec:binary-rs-list}

Optimistically, we might hope to show a close deterministic analog of the probabilistic result from \cref{lem:coset-sampling}, namely, an efficient deterministic construction of a parity-check matrix $H=H_q(k,S)$ and a shift $\vec{x} \in \Z^S$ where the coset $\vec{x} + \latpar(H)$ has sufficiently many \emph{binary} vectors of some specified Hamming weight
$h \leq \alpha \cdot (2k) \leq \alpha \cdot \lambda_1^{(1)}(\latpar(H))$ for some positive constant $\alpha < 1$.

However, even with a relaxed requirement of $h \leq 2k$, the hoped-for deterministic analog would imply explicit Reed-Solomon list-decoding configurations that go beyond the current state of the art, due to Guruswami and Rudra~\cite{DBLP:journals/tit/GuruswamiR06}.
While this does not rule out the desired analog entirely, it does present an apparent barrier:
achieving the analog would likely require new techniques, and would lead to a significant step forward in our understanding of Reed-Solomon codes.

\begin{lemma}
\label{lem:binary-rs-list}
There is an efficient deterministic algorithm that, given a prime~$q$, a positive integer~$k \leq q$, a set $S \subseteq \F_q$ of cardinality $\card{S} \geq k$ defining $H=H_q(k,S)$ and $\lat=\latpar(H) \subseteq \Z^S$, an integer $h \in \set{k, \ldots, \card{S}}$, and any shift $\vec{x} \in \Z^S$ (or its syndrome $\vec{u}=H\vec{x} \in \F_q^k$), outputs a vector $\vec{r} \in \F_q^S$ for which there are at least $\card{(\vec{x}+\lat) \cap B_{S,h})}$ codewords in the Reed-Solomon code $\RS_q[h-k+1, S]$ that each agree with~$\vec{r}$ in at least~$h$ coordinates.
\end{lemma}

\begin{proof}
The algorithm works as follows.
\begin{enumerate}[itemsep=0pt]
    \item Compute the coset's syndrome $\vec{u} = H \vec{x} \in \F_q^k$.

    Note that for each $i \in [k]$, the syndrome's $i$th coordinate $u_i$ is the $i$th power sum of the subset $T \subseteq S$ indicated by any binary vector $\vec{x}_T \in (\vec{x}+\lat) \cap B_{S,h}$, i.e., $u_i = (H\vec{x})_i = (H \vec{x}_T)_i = p_i(T) \in \F_q$.
    In particular, $h = \card{T} \equiv u_0 \pmod{q}$ for any such~$T$.
    (So, if $h \not\equiv u_0 \pmod{q}$, no such~$T$ exists and the algorithm may output any vector.)

    \item \label{item:elesymm-from-newton} Use Newton's identities (\cref{eq:newtons-identities}) to compute (from~$\vec{u}$) the values $s_i = e_i(T) \in \F_q$ for each $i \in [k]$, which are the first $k \leq h = \card{T}$ elementary symmetric polynomials of any such~$T$.

    \item Finally, construct the polynomial
    \[ r(x) = \sum_{i \in [k]} (-1)^i \cdot s_i \cdot x^{h-i} \in \F_q[x] \]
    and output the vector $\vec{r} = (-r(s))_{s \in S} \in \F_q^S$ obtained by evaluating~$-r(x)$ at each element of~$S$.

\end{enumerate}

We now show that the algorithm is correct.
Let $T \subseteq S$ be the subset indicated by a binary vector $\vec{x}_T \in (\vec{x}+\lat) \cap B_{S,h}$, and let $f_T(x) = \prod_{t \in T} (x-t) \in \F_q[x]$ be the root polynomial of $T$.
As argued above in \cref{item:elesymm-from-newton}, $s_i = e_i(T)$ for each $i \in [k]$, so by \cref{eq:root-polynomial}, we have
\[ f_T(x) = r(x) + g_T(x) \]
for some $g_T(x)$ where $\deg(g_T) \leq h-k$ (or~$g_T$ is identically zero).
Because $f_T(t) = 0$ for every $t \in T$, we have $g_T(t) = -r(t)$ for every $t \in T$, i.e., the Reed-Solomon codeword in $\text{RS}[h-k+1, S]$ defined by~$g_T$ agrees with~$\vec{r}$ in at least~$\card{T}=h$ positions.
Finally, the claim follows by noting that each distinct binary vector $\vec{x}_T \in (\vec{x}+\lat) \cap B_{S,h}$ yields a distinct subset~$T \subseteq S$, which yields a distinct root polynomial~$f_T$ and thereby a distinct polynomial $g_T = f_T - r$.
\end{proof}

For an analog of \cref{lem:coset-sampling}, we want an explicit (i.e., deterministically and efficiently constructible) lattice $\lat = \latpar(H)$ and shift~$\vec{x}$ such that $\card{(\vec{x}+\lat) \cap B_{S,h}}$ is super-polynomial (in~$\card{S}$), for some Hamming weight $h \leq 2k \leq \lambda_1^{(1)}(\lat)$.
By \cref{lem:binary-rs-list}, this implies an explicit received word and Reed-Solomon code of dimension $h-k+1$ for which super-polynomially many codewords agree with the received word in at least~$h$ coordinates.
This is an agreement-to-dimension ratio of $h/(h-k+1) = 1 + (k-1)/(h-k+1) \geq 2 - O(1/k)$.
However, the state of the art for explicit Reed-Solomon list-decoding configurations~\cite[Corollary~2]{DBLP:journals/tit/GuruswamiR06} requires a ratio of $2-\Omega(1)$ in order to get a super-polynomial list size.

\begin{remark}
It is not clear whether a useful converse of \cref{lem:binary-rs-list} holds, i.e., whether an explicit `bad' list-decoding configuration for Reed-Solomon
codes would yield a corresponding explicit locally dense lattice.
This is because the difference between the received word and a nearby codeword may not correspond to (the evaluations of) the root polynomial of the positions at which they agree, but is only divisible by it.
\end{remark}

\begin{remark}
\label{rem:generalize-rs-list}
Everything in the proof of \cref{lem:binary-rs-list} generalizes from binary vectors to \emph{non-negative} integer vectors of~$\ell_1$ norm~$h$ and their corresponding \emph{multisets}~$T$, except for the final step:
the codeword corresponding to $g_T(x)$ is guaranteed to agree with~$\vec{r}$ only at each \emph{distinct} $t \in T$, and the number of these is only the \emph{Hamming} weight (not the~$\ell_1$ norm~$h$) of the integer vector.
It appears that this gap can be overcome (restoring agreement~$h$) by instead using univariate \emph{multiplicity codes}, which are generalizations of Reed-Solomon where the coordinates correspond to `multiple evaluations' of a polynomial \emph{and its derivatives}; see~\cite{Kopparty14-Multiplicity}.
For such codes, we do not know whether any explicit list-decoding configurations with substantially different parameters from Reed-Solomon codes are known.

By contrast, \cref{lem:binary-rs-list} does not seem to generalize easily to \emph{arbitrary} integer vectors that may have negative entries (e.g., ternary $\set{0,\pm 1}$-vectors of a given Hamming weight).
So, considering negative entries may provide a route to constructing locally dense lattices without having to overcome the Reed-Solomon list-decoding barriers.
\end{remark}

\subsection{A Fourier-Analytic Approach}
\label{sec:fourier}

We next outline an approach for using Fourier analysis over $\F_q^k$ for prime~$q$ and $k \ll q$ to show that an explicit coset $\vec{x} + \latpar(H)$ for $H = H_q(k, \F_q))$ contains many vectors of $\ell_1$ norm at most $h \leq \alpha \cdot (2 k) \leq \alpha \cdot \lambda_1^{(1)}(\latpar(H))$ for some integer $h$, where $\alpha \in (0, 1)$ is a constant as above.

The approach uses techniques similar to those used by Cheng and Wan~\cite{conf/focs/ChengW04} and Guruswami and Rudra~\cite{DBLP:journals/tit/GuruswamiR06} for constructing explicit Hamming balls containing many Reed-Solomon codewords (i.e., for constructing explicit ``bad'' configurations for list-decoding Reed-Solomon codes), and by Cheng and Wan~\cite{journals/tit/ChengW12} for showing deterministic hardness of the minimum distance problem on codes.

Specifically, we perform a Fourier-analytic calculation similar to ones appearing in those works in an attempt to lower bound the number of short vectors $\vec{x} \in \F_q^q$ satisfying $H\vec{x} = \vec{s}$ for some syndrome $\vec{s} \in \F_q^k$ (which is equivalent to lower bounding the number of short vectors in the lattice coset $\latpar_{\vec{s}}(H)$). 
The calculation starts by establishing a correspondence between polynomials in $\F_q[x]$ of degree less than~$k$ and the Fourier coefficients of a certain indicator function~$1_W$. It then uses a combination of the Weil bound (\cref{eq:weil-additive}) and the triangle inequality to upper bound the total contribution of the terms corresponding to non-constant polynomials in the Fourier expansion of the $h$-fold convolution $1_W^{(\ast h)}$ of~$1_W$, which essentially counts the number of non-negative integer vectors of~$\ell_1$ norm~$h$ with a given syndrome (see \cref{eq:weil-triangle-argument}).

Unfortunately, this calculation (just barely) fails to show anything useful for our purposes. The rough idea of why is as follows.
The argument seeks to show that the sum of magnitudes ($\ell_1$ norm) of all the Fourier coefficients of $1_W^{(\ast h)}$ corresponding to non-constant polynomials is less than the zero coefficient of $1_W^{(\ast h)}$ (which is real-valued, positive, and corresponds to the zero polynomial).
The Weil bound shows that there is a multiplicative gap of about $\sqrt{q}$ between the value of the zero coefficient and the magnitude of any non-constant coefficient of~$1_W$, and this gap is amplified to about $q^{h/2}$ for the corresponding Fourier coefficients of $1_W^{(\ast h)}$.
However, because there are about~$q^k$ non-constant polynomials in $\F_q[x]$ of degree less than~$k$, we need to take $h \geq 2k$ in order for the zero coefficient of $1_W^{(\ast h)}$ to be large enough, which conflicts with the requirement that $h < 2k$.

Because the calculation is not directly useful for our purposes, and due to its similarity to those in prior work, we defer it to \cref{sec:app-fourier}.
There we also discuss some approaches for making the argument work.

\subsection{Point-Counting Proxies}
\label{sec:point-counting}

In this section we outline a possible approach toward proving local density for a deterministic version of our construction, using ``smooth'' functions that serve as proxies for point-counting functions.
The hope is to use these functions to prove a lower bound on the number of sufficiently short points in \emph{any} coset $\vec{x}+\latpar(H)$ meeting some easy-to-ensure conditions.
The approach is centered on analytical techniques originally due to~\cite{mazo90:_lattic_point_high_dimen_spher,elkies91:_packing_densities}, which were successfully used on the integer lattice $\Z^n$ for complexity-theoretic purposes in~\cite{conf/stoc/AggarwalS18,conf/coco/BennettP20}.
In the present context of Reed-Solomon lattices, our attempts have been inconclusive: while we did not manage to prove what we are seeking, we also do not see any inherent barrier, and further effort may lead to a favorable outcome.

For any $p \geq 1$, $\ell_p$ norm bound $r \geq 0$, and lattice coset $\vec{x}+\lat \subset \R^n$, define
\begin{equation}
    N_p(r; \vec{x}+\lat) := \card{(\vec{x}+\lat) \cap \B_p^n(r)}
\end{equation}
to be the number of points in $\vec{x}+\lat$ of~$\ell_p$ norm at most~$r$.
Similarly, for any $\tau > 0$, define the function
\begin{equation}
    \label{eq:Theta_p}
    \Theta_p(\tau; \vec{x}+\lat) := \sum_{\vec{v} \in \vec{x}+\lat} \exp(-\tau \norm{\vec{v}}_p^p) \ \text{.}
\end{equation}
In what follows, for brevity we often omit the coset $\vec{x}+\lat$ when it is fixed and clear from context.

First, for any $\tau > 0$ we immediately have
\begin{equation}
    \label{eq:N_p-upper}
    N_p(r) \leq \exp(\tau \cdot r^p) \cdot \Theta_p(\tau) \ \text{,}
\end{equation}
because 
each point counted by $N_p(r)$ contributes at least $\exp(-\tau \cdot r^p)$ to $\Theta_p(\tau)$, and all other points contribute positively.
The work of~\cite{mazo90:_lattic_point_high_dimen_spher,elkies91:_packing_densities} showed that for any coset of the \emph{integer} lattice $\lat = \Z^n$, the upper bound from \cref{eq:N_p-upper} is fairly tight (up to $\exp(O(\sqrt{n}))$ factors) for suitably chosen~$\tau$, namely, the one for which $\mu_p(\tau; \vec{x}+\lat) = r^p$, where
\begin{equation}
    \mu_p(\tau; \vec{x}+\lat) := \E_{\vec{v} \sim D_p(\tau; \vec{x}+\lat)} \bracks*{\norm{\vec{v}}_p^p}
\end{equation}
and $D_p(\tau)$ is the probability distribution supported on $\vec{x}+\lat$ that assigns probability $\exp(-\tau \norm{\vec{v}}_p^p)/\Theta_p(\tau)$ to each~$\vec{v}$.
In other words, $\mu_p(\tau)$ is the $p$th moment of the~$\ell_p$ norm of a sample from $D_p(\tau)$, and when~$\tau$ is chosen to make this moment equal~$r^p$, \cref{eq:N_p-upper} is fairly tight for any coset of the integer lattice.

Optimistically, if we could show that \cref{eq:N_p-upper} is fairly tight for suitable~$\tau$ (as a function of~$r$) and every coset of our special lattices, then we would have
\[ N_p(r) \approx \exp(\tau \cdot r^p) \cdot \Theta_p(\tau) = \exp(\tau \cdot \alpha^p \cdot 2k) \cdot \Theta_p(\tau) \ \text{,} \]
where $r = \alpha \cdot (2k)^{1/p} \leq \alpha \cdot \lambda_1^{(p)}(\latpar(H))$ for some constant $\alpha < 1$ is our distance of interest for local density, and~$\tau$ is determined by this~$r$.
Furthermore, if we could also show that, say, $\Theta_p(\tau) \geq \exp(-\tau k)$ for this~$\tau$ and a suitable choice of coset (e.g., by shifting the lattice by some~$\vec{x}$ with $\norm{\vec{x}}_p^p=k$), then for any constant $\alpha > 1/2^{1/p}$ we could conclude that there are $\exp(\Omega(\tau k))$ sufficiently short points in the coset.
Finally, if we could show that $\tau \geq k^{c-1}$ for some constant $c > 0$, this would yield a subexponential number of short coset vectors, as needed.

\paragraph{An approach.}

Toward this end, we observe that several parts of the analysis from~\cite{mazo90:_lattic_point_high_dimen_spher,elkies91:_packing_densities} immediately generalize to cosets of arbitrary lattices, not just the integer lattice~$\Z^n$.
In particular, for any lattice coset and any $\tau > 0$, by a routine calculation we have
\begin{align}
    \frac{\partial}{\partial \tau} \ln \Theta_p(\tau) &= -\mu_p(\tau) < 0 \ \text{,}
    \label{eq:partial-Theta_p} \\
    \frac{\partial^2}{\partial \tau^2} \ln \Theta_p(\tau) &=
    \E_{\vec{v} \sim D_p(\tau)} \bracks*{\norm{\vec{v}}_p^{2p}} - \mu_p(\tau)^2 > 0 \ \text{.} \label{eq:partial-partial-Theta_p}
\end{align}
In words, $\ln \Theta_p(\tau)$ is decreasing and concave up, and its second partial derivative (with respect to~$\tau$) is the variance of $\norm{\vec{v}}_p^p$ when~$\vec{v}$ is sampled from $D_p(\tau)$.
Moreover, by an elegant (and non-obvious) argument, for any $\tau, \delta > 0$ we have the lower bound
\begin{equation}
    \label{eq:N_p-lower}
    N_p(\mu_p(\tau)^{1/p}) \geq \exp((\tau+\delta) \cdot \mu_p(\tau + 2\delta)) \cdot H_p(\tau, \delta) \ \text{,}
\end{equation}
where
\begin{equation}
    H_p(\tau, \delta) := \Theta_p(\tau+\delta) - \exp(-\delta \mu_p(\tau)) \cdot \Theta_p(\tau) - \exp(\delta \mu_p(\tau+2\delta)) \cdot \Theta_p(\tau+2\delta) \ \text{.}
\end{equation}

Using the Taylor series for $\ln \Theta_p(\tau)$ around~$\tau$ and
\cref{eq:partial-Theta_p}, for any $\delta > 0$ we have
\[
    \ln \Theta_p(\tau) \leq \ln \Theta_p(\tau+\delta) + \delta \mu_p(\tau) - \frac{\delta^2}{2} \inf_{\tau' \in [\tau,\tau+\delta]} \frac{\partial^2}{\partial \tau^2} \ln \Theta_p(\tau') \ \text{,}
\]
and similarly for $\ln \Theta_p(\tau+2\delta)$.
So, if we could suitably lower-bound the above infimum on the second partial derivative (which, to recall from \cref{eq:partial-partial-Theta_p}, is positive) for appropriate~$\delta$, by taking $\exp(\cdot)$ of both sides and rearranging we could show that, say, $H_p(\tau,\delta) \geq \Theta_p(\tau+\delta)/2$.
The work of~\cite{mazo90:_lattic_point_high_dimen_spher,elkies91:_packing_densities} actually does this for the integer lattice~$\Z^n$, by exploiting its product structure to show that the infimum is proportional to~$n$, but here the path forward is less clear.

If the above could be achieved, then we would get the lower bound
\[ N_p(\mu_p(\tau)^{1/p})
\geq \exp((\tau+\delta) \cdot \mu_p(\tau+2\delta)) \cdot \Theta_p(\tau+\delta)/2
\approx \exp(\tau \cdot \mu_p(\tau)) \cdot \Theta_p(\tau)/2 \ \text{,} \]
showing that \cref{eq:N_p-upper} is fairly tight when $r=\mu_p(\tau)^{1/p}$.
As explained above, the final goal would then be to identify an explicit coset for which the right-hand side is at least subexponentially large in the lattice dimension, when~$\tau$ is set so that $\mu_p(\tau) = \alpha^p \cdot (2k)$ for some constant $\alpha < 1$.

%% file: app-svpcvp-reduction.tex
\section{Proof of \texorpdfstring{\cref{thm:cvp-to-svp}}{CVP to SVP reduction theorem}}
\label{sec:cvp-svp-reduction-proof}

Here we restate and prove \cref{thm:cvp-to-svp}.

\cvpsvp*

\begin{proof}
Let $B \in \Z^{d \times r}$, $\vec{t} \in \Z^d$, $s > 0$ be the input instance of $\gamma$-$\GapCVP_p'$.
Assume without loss of generality that $\vec{t} \neq \vec{0}$ (if $\vec{t} = \vec{0}$, simply output an arbitrary YES instance of $\gamma'$-$\GapSVP_p$).
The reduction outputs $(B', s')$, where
\begin{align*}
B' &:= \begin{pmatrix}
BTA &  BT \vec{x} - \vec{t} \\
\beta A & \beta \vec{x}
\end{pmatrix}
\text{ with }
\beta := \frac{\gamma' s}{\ell^{1/p} \cdot (1 - (\alpha \gamma')^p)^{1/p}} \neq 0
\ \text{,} \\
s' &:= (s^p + \alpha^p \beta^p \ell)^{1/p} \ \text{.}
\end{align*}

It is clear that the reduction runs in polynomial time.
To show correctness, we start by showing that the columns of $B'$ are linearly independent, i.e., that $B'$ is a lattice basis, and hence the lattice it generates has rank $R+1$.
Recall that $A \in \Z^{n \times R}$ (for some~$n$) is a lattice basis, so it has linearly independent columns.
Let $\vec{u} \in \R^{R + 1} \setminus \set{\vec{0}}$.
If $(\beta A \mid \beta \vec{x}) \cdot \vec{u} \neq \vec{0}$ then $B' \vec{u} \neq \vec{0}$.
Otherwise, $(A \mid \vec{x}) \cdot \vec{u} = \vec{0}$ and so $u_{n+1} \neq 0$ by linear independence of $A$'s columns. Then
\[ (B T A \mid B T \vec{x} - \vec{t}) \cdot \vec{u} = BT (A \mid \vec{x}) \cdot \vec{u} - u_{n+1} \cdot \vec{t} = -u_{n+1} \cdot \vec{t} \neq \vec{0} \ \text{,} \] because $\vec{t} \neq \vec{0}$ by assumption. Hence $B' \vec{u} \neq \vec{0}$, as needed.

We now show the main claim about $(B', s')$. %
First, suppose that the input $\GapCVP'_p$ instance $(B, \vec{t}, s)$ is a YES instance.
Then by \cref{def:gapcvp-prime} there exists an $\vec{c} \in \bit^r$ such that $\norm{B\vec{c} - \vec{t}}_p \leq s$.
Moreover, by \cref{item:cover-hypercube} of \cref{def:locally-dense}, there exists $\vec{v} = \vec{x} + A\vec{z} \in \vec{x} + \lat(A)$ for some $\vec{z} \in \Z^{R}$ such that $\norm{\vec{v}}_p^p \leq \alpha^p \cdot \ell$ and $T\vec{v} = \vec{c}$. For such a $\vec{z}$, we have $T(\vec{x} + A\vec{z}) = T\vec{v} = \vec{c}$
and therefore
\begin{align*}
\lambda_1^{(p)}(\lat(B'))^p &\leq \norm{B'\cdot (\vec{z}, 1)}_p^p \\
&= \norm{BTA\vec{z} + BT\vec{x} - \vec{t}}_p^p + \beta^p \norm{A\vec{z} + \vec{x}}_p^p \\
&= \norm{B\vec{c} - \vec{t}}_p^p + \beta^p \norm{A\vec{z} + \vec{x}}_p^p \\
& \leq s^p + \alpha^p \beta^p \ell \\
&= (s')^p \ \text{,}
\end{align*}
as needed.

Next, suppose that the input $\GapCVP'_p$ instance is a NO instance.
We will show that $\lambda_1^{(p)}(\lat(B')) \geq \gamma' s'$ by showing that $\norm{B' \cdot (\vec{z}, w)}_p^p > (\gamma' s')^p$ for all $(\vec{z}, w) \in \Z^{R + 1} \setminus \set{\vec{0}}$, separately analyzing the cases $w = 0$ and $w \neq 0$.
When $w = 0$, we have that $\vec{z} \neq \vec{0}$, and so by \cref{item:min-dist} of \cref{def:locally-dense},
\[ \norm{B' \cdot (\vec{z}, w)}_p^p
\geq \beta^p \cdot \lambda_1^{(p)}(\lat(A))^p
> \beta^p \cdot \ell \ \text{.} \]
And, since $\beta^p \cdot \ell \cdot (1-(\alpha \gamma')^p) = (\gamma' s)^p$ by definition, we have (by definition of~$s'$)
\[ \beta^p \cdot \ell \geq (\gamma' s)^p + (\alpha \beta \gamma')^p \cdot \ell = (\gamma' s')^p \ \text{.} \]

When $w \neq 0$, by the definition of $\GapCVP'_p$ NO instances, we have
\[ \norm{B' \cdot (\vec{v}, w)}_p^p \geq \dist_p(w \vec{t}, \lat(B))^p > (\gamma s)^p \ \text{.} \]
Then, by the definitions of~$\gamma$,~$\beta$, and~$s'$, we have 
\[ (\gamma s)^p \geq (\gamma')^p \cdot \frac{s^p}{1-(\alpha \gamma')^p}
= (\gamma')^p \cdot \parens*{s^p + \frac{(\alpha \gamma' s)^p}{1-(\alpha \gamma')^p}}
= (\gamma')^p \cdot (s^p + \alpha^p \beta^p \ell)
= (\gamma' s')^p \ \text{.} \] 
The claim follows.
\end{proof}

%% file: app-fourier.tex
\section{A Fourier-Analytic Approach to Derandomization}
\label{sec:app-fourier}

In this appendix we perform and discuss the Fourier-analytic calculation described in \cref{sec:fourier}.
We introduce definitions and basic facts briefly, and refer the reader to \cite{KowalskiSumsFiniteFields} and \cite[Chapter 8]{odonnell/boolean} for details.
The \emph{additive Fourier characters} of $\F_q^k$ are the homomorphisms from the additive group $\Z_q^k$ of $\F_q^k$ to the multiplicative group of $\mathbb{C}$. Concretely, they are the $q^k$ functions $\Psi_{\vec{u}}(\vec{x}) := \omega^{\iprod{\vec{u}, \vec{x}}}$
indexed by $\vec{u} \in \F_q^k$, where $\omega$ is an arbitrary fixed $q$th primitive root of unity (e.g., $\omega = \exp(2 \pi i / q)$).
For a function $f: \F_q^k \to \mathbb{C}$, its $\vec{u}$th Fourier coefficient $\fh(\vec{u})$ for $\vec{u} \in \F_q^k$ is defined as
\[
\fh(\vec{u}) := \E_{\vec{b} \sim \F_q^k} [f(\vec{b}) \overline{\Psi_{\vec{u}}(\vec{b})}] \ \text{,}
\]
and the Fourier expansion of $f$ is
\[
f(\vec{x}) = \sum_{\vec{u} \in \F_q^k} \fh(\vec{u}) \Psi_{\vec{u}}(\vec{x}) \ \text{.}
\]

Let $W = W(k, q) := \set{(1, a, a^2, \ldots, a^{k - 1})^T : a \in \F_q} \subseteq \F_q^k$ denote the set of columns of the matrix $H_q(k, \F_q)$,
and let $1_W : \F_q^k \to \bit$ denote the indicator function of~$W$.
The Fourier coefficients of~$1_W$ satisfy
\[
\widehat{1_W}(\vec{u}) = \E_{\vec{b} \sim \F_q^k} [1_W(\vec{b}) \overline{\Psi_{\vec{u}}(\vec{b})}] 
= \frac{1}{q^k} \sum_{\vec{b} \in W} \overline{\Psi_{\vec{u}}(\vec{b})}
= \frac{1}{q^k} \sum_{a \in \F_q} \Psi(-p_{\vec{u}}(a)) \ \text{,}
\]
where $\Psi(x) = \omega^x$ is an additive character of $\F_q$ and $p_{\vec{u}}(x) = \sum_{i=0}^{k-1} u_i x^i \in \F_q[x]$ is a polynomial of degree less than $k \leq q$.%
\footnote{Note that formally, we are indexing the Fourier coefficients by $\F_q^{[k]}$, i.e., the coordinates of $\vec{u} \in \F_q^{[k]}$ are indexed from zero, by the elements of $[k] = \set{0, \ldots, k - 1}$.}
For non-constant polynomials $p_{\vec{u}}(x)$, the Weil bound for additive character sums (see~\cite[Theorem 3.2]{KowalskiSumsFiniteFields}) asserts that 
\begin{equation}
\label{eq:weil-additive}
q^k \cdot \abs{\widehat{1_W}(\vec{u})}
= \abs[\Big]{\sum_{a \in \F_q} \Psi(-p_{\vec{u}}(a))}
\leq (k - 2) \sqrt{q} \ \text{.}
\end{equation}

The \emph{convolution} of two functions $f, g : \F_q^k \to \mathbb{C}$ is the function defined as $(f \ast g)(\vec{x}) := \E_{\vec{y} \sim \F_q^k}[f(\vec{x} - \vec{y})g(\vec{y})]$.
Importantly, Fourier coefficients are multiplicative under convolution:
\[ \widehat{(f \ast g)}(\vec{u}) = \fh(\vec{u})\gh(\vec{u}) \ \text{.} \]
Furthermore, the $h$-fold convolution of $1_W(\vec{x})$ with itself, denoted $1_W^{(\ast h)}(\vec{x})$, counts, up to a normalization factor of~$q^{hk}$, the number of sequences $(\vec{w}_1, \ldots, \vec{w}_h) \in W^h$ such that $\sum_{i = 1}^h \vec{w}_i = \vec{x}$.%
\footnote{\label{foot:sequences-vs-multisets} Technically, we care about the number of \emph{multisets} (rather than sequences) of vectors $V \subseteq W$, $\card{V} = h$ such that $\sum_{\vec{w} \in V} \vec{w} = \vec{x}$. Indeed, there is a bijection between such multisets and their indicator vectors $\vec{y}' \in (\Z^{\geq 0})^q$, which are such that $\norm{\vec{y}'}_1 = h$ and $H\vec{y}' = \vec{x}$. However, the number of such multisets is at most an $h!$ factor less than the number of such sequences, and this factor seems inconsequential for our purposes; see \cref{eq:weil-triangle-argument} and the subsequent discussion.}

Let $U_0 := \set{(u_0, 0, \ldots, 0) : u_0 \in \F_q}$ denote the set of coefficient vectors corresponding to constant polynomials, i.e., corresponding to polynomials $p_{\vec{u}}(x) = u_0$ for $u_0 \in \F_q$.
Putting everything together, we get that
\begin{align*}
\frac{1}{q^{hk}} \cdot \abs[\Big]{\set{(\vec{w}_1, \ldots, \vec{w}_h) \in W^h : \sum_{i=1}^h \vec{w}_i = \vec{x}}}
&= 1_W^{(\ast h)}(\vec{x}) \\
&= \sum_{\vec{u} \in \F_q^k} \widehat{1_W^{(\ast h)}}(\vec{u}) \cdot \Psi_{\vec{u}}(\vec{x}) \\
&= \sum_{\vec{u} \in \F_q^k} (\widehat{1_W}(\vec{u}))^h \cdot \Psi_{\vec{u}}(\vec{x}) \\
&= \sum_{\vec{u} \in U_0} (\widehat{1_W}(\vec{u}))^h \cdot \Psi_{\vec{u}}(\vec{x}) + \sum_{\vec{u} \in \F_q^k \setminus U_0} (\widehat{1_W}(\vec{u}))^h \cdot \Psi_{\vec{u}}(\vec{x}) \\
&= \frac{1}{q^{hk}} \sum_{u_0 \in \F_q} (q \cdot \Psi(-u_0))^h \cdot \Psi(u_0 x_0) 
+ \sum_{\vec{u} \in \F_q^k \setminus U_0} (\widehat{1_W}(\vec{u}))^h \cdot \Psi_{\vec{u}}(\vec{x}) \\
&= \frac{1}{q^{h (k-1)}} \sum_{u_0 \in \F_q} \Psi(u_0 (x_0 - h)) 
+ \sum_{\vec{u} \in \F_q^k \setminus U_0} (\widehat{1_W}(\vec{u}))^h \cdot \Psi_{\vec{u}}(\vec{x}) \ \text{.}
\end{align*}
Now let $\vec{s} \in \F_q^k$ be such that $s_0 = h \bmod q$.
Multiplying the above equality by~$q^{hk}$, we have
\begin{align}
\label{eq:weil-triangle-argument}
\begin{split}
\Big|\set{(\vec{w}_1, \ldots, \vec{w}_h) \in W^h : \sum_{i=1}^h \vec{w}_i = \vec{s}}\Big| &= 
q^{h + 1} + q^{hk} \sum_{\vec{u} \in \F_q^k \setminus U_0} (\widehat{1_W}(\vec{u}))^h \cdot \Psi_{\vec{u}}(\vec{s}) \\
& \geq q^{h + 1} - q^{hk} \sum_{\vec{u} \in \F_q^k \setminus U_0} \Big| (\widehat{1_W}(\vec{u}))^h \cdot \Psi_{\vec{u}}(\vec{s}) \Big| \\
& \geq q^{h + 1} - q^k \cdot ((k - 2) \sqrt{q})^h \ \text{,}
\end{split}
\end{align}
where the equality uses the fact that $s_0 = h \bmod{q}$, the first inequality uses the triangle inequality, and the second inequality uses the Weil bound (\cref{eq:weil-additive}).
Taking $k = q^{\eps'}$ for some constant $\eps' > 0$, the last bound is roughly $q^{h} - q^k \cdot q^{(1/2 + \eps')h} = q^{h} - q^{k + (1/2 + \eps')h}$, which we want to be subexponentially large in $q$.

However, in order to even make the above quantity positive, we must take $h > k/(1/2 - \eps') = (2 + \eps)k$ for some constant $\eps > 0$, whereas our application needs $h < 2k$ (since we know only that $\lambda_1^{(1)}(\latpar(H)) \geq 2k$).
Indeed, this same condition shows up in related work: our requirement that $h > (2 + \eps)k$ is essentially the same as the requirement that ``$g \geq (2 + \eps) (h + 1)$'' in the premise of \cite[Theorem 2.1]{journals/tit/ChengW12}, even though the setting there is slightly different.
In particular, in each case the ``$2$'' in the $2 + \eps$ factor comes from (the inverse of) the exponent in the Weil bound's $q^{1/2}$ factor.

\paragraph{Toward making the approach work.}

The above argument `just barely' fails---i.e., improving the constant from $2+\eps$ to $2-\eps$ would suffice for our purposes---so one might hope to make it work by showing that one of the inequalities in \cref{eq:weil-triangle-argument} is loose.
First, one might hope to strengthen the Weil bound by reducing its~$\sqrt{q}$ factor to $q^{1/2 - \delta}$ for some constant $\delta > 0$.
Furthermore, the Weil bound is a ``worst-case'' statement about character sums evaluated on \emph{arbitrary} polynomials, whereas for the argument above an improvement to $q^{1/2 - \delta}$ on average (i.e., for \emph{random} polynomials) would suffice. Unfortunately, as shown in~\cite{Sawin/Weil}, the~$\sqrt{q}$ factor is optimal even for such an ``average-case Weil bound.''
So, the main possibility for improvement is to avoid na\"{i}vely using the triangle inequality in the first inequality.
Indeed, there could be substantial ``phase cancellations'' among the terms in the sum $\sum_{\vec{u} \in \F_q^k \setminus U_0} (\widehat{1_W}(\vec{u}))^h \cdot \Psi_{\vec{u}}(\vec{s})$ for some special (or perhaps arbitrary) syndromes $\vec{s}$.
It seems likely that such cancellations occur, but we do not know how to show it.
(We emphasize that \cref{eq:weil-triangle-argument} holds for \emph{any}~$\vec{s}$ with $s_0=h \bmod q$; it is agnostic to the rest of~$\vec{s}$.) 

Finally, we note another possible modification to the argument, which is to work with the ``symmetrized set'' $W' := \pm W = \set{\pm (1, a, a^2, \ldots, a^{k-1})^T : a \in \F_q}$ instead of~$W$.
A multiset of~$h$ vectors from~$W'$ naturally corresponds to a ``signed'' vector $\vec{y}' \in \Z^q$, whereas such a multiset of vectors in~$W$ corresponds to an ``unsigned'' vector $\vec{y} \in (\Z^{\geq 0})^q$.
That is, we can count coset vectors having negative coordinates when working with~$W'$ instead of~$W$.
However, it is unclear to us how to use this relaxation to any advantage.